\newif\ifdraft
\draftfalse

\documentclass[
  10pt,
  letterpaper,
  headings=small,
  numbers=noendperiod,
  bibliography=totoc
]{scrartcl}

\usepackage[
  left=1.00in,
  right=1.00in,
  top=0.86in,
  bottom=0.95in,
  includehead
]{geometry}
\usepackage[T1]{fontenc}
\usepackage[utf8]{inputenc}
\usepackage{textcomp}
\usepackage{microtype}
\usepackage[tt=false,type1=true]{libertine}

\usepackage{amsmath,amsthm,mathtools}
\usepackage[libertine]{newtxmath}
\usepackage{bm}
\usepackage{dsfont}
\usepackage{xfrac}
\usepackage{xspace}
\usepackage{xargs}

\theoremstyle{plain}
\newtheorem{theorem}{Theorem}
\newtheorem{lemma}[theorem]{Lemma}
\newtheorem{proposition}[theorem]{Proposition}
\newtheorem{corollary}[theorem]{Corollary}

\theoremstyle{definition}
\newtheorem{Definition}[theorem]{Definition}

\theoremstyle{remark}

\renewcommand{\Pr}[1]{\mathds{P}\!\left[#1\right]}

\renewcommand{\epsilon}{\varepsilon}

\usepackage{graphicx}
\usepackage{booktabs}
\usepackage{array}
\usepackage{tabularx}
\usepackage{longtable}
\usepackage{multirow}
\usepackage{wrapfig}
\usepackage{float}
\usepackage{pdflscape}
\usepackage[font=small,labelfont=bf]{caption}
\usepackage[ruled,vlined]{algorithm2e}

\SetAlFnt{\small}
\SetAlCapFnt{\small}
\SetAlCapNameFnt{\small}
\SetAlCapHSkip{0pt}
\IncMargin{-\parindent}

\usepackage[shortlabels]{enumitem}
\usepackage{xcolor}
\definecolor{darkred}{rgb}{0.5,0,0}
\definecolor{lightblue}{rgb}{0,0.4,0.8}
\definecolor{darkgreen}{rgb}{0,0.5,0}
\definecolor{linkblue}{rgb}{0.05,0.20,0.45}

\usepackage[numbers,sort&compress]{natbib}
\usepackage[
  colorlinks=true,
  linkcolor=black,
  citecolor=linkblue,
  urlcolor=linkblue
]{hyperref}
\usepackage{url}

\usepackage[automark]{scrlayer-scrpage}
\clearpairofpagestyles
\ihead{\footnotesize Decision-Making Under Complete Uncertainty: You Will Not Regret Being Greedy}
\ohead{\footnotesize\pagemark}
\setkomafont{title}{\normalfont\bfseries\Large}
\setkomafont{author}{\normalfont\normalsize}
\setkomafont{date}{\normalfont\normalsize}
\setkomafont{disposition}{\normalfont\sffamily\bfseries}
\RedeclareSectionCommand[beforeskip=1.7ex plus .4ex minus .2ex,afterskip=.75ex plus .2ex]{section}
\RedeclareSectionCommand[beforeskip=1.3ex plus .3ex minus .2ex,afterskip=.45ex plus .15ex]{subsection}
\RedeclareSectionCommand[beforeskip=1.0ex plus .25ex minus .15ex,afterskip=.35ex plus .1ex]{subsubsection}

\makeatletter
\renewcommand{\maketitle}{%
  \begingroup
  \thispagestyle{plain}%
  \raggedright
  {\Large\bfseries \@title\par}%
  \vspace{0.70em}%
  {\normalsize \@author\par}%
  \ifx\@date\@empty\else
    \vspace{0.25em}{\normalsize \@date\par}%
  \fi
  \vspace{0.80em}%
  \endgroup
}

\newcommand{\compacttoc}{%
  \par\vspace{0.65em}%
  \begin{center}{\sffamily\bfseries Contents}\end{center}%
  \vspace{-0.80em}%
  \@starttoc{toc}%
}
\makeatother

\renewenvironment{abstract}{%
  \par\vspace{.65em}%
  \begin{center}{\sffamily\bfseries Abstract}\end{center}%
  \vspace{-0.55em}%
  \small\noindent\ignorespaces%
}{\par\normalsize\vspace{.45em}}

\newenvironment{keywords}{%
  \par\smallskip\noindent{\small\bfseries Additional Key Words and Phrases: }\small\ignorespaces%
}{\par\normalsize\vspace{.45em}}

\title{Decision-Making Under Complete Uncertainty: You Will Not Regret Being Greedy}
\author{%
	Kristijan Atanasov\textsuperscript{1}, Frederik Mallmann-Trenn\textsuperscript{1}, and Mehmet Mars Seven\textsuperscript{2} 
}

\date{\footnotesize\textsuperscript{\textbf{1}}Department of Informatics, King's College London\\ \textsuperscript{\textbf{2}}Department of Political Economy, King's College London}

\begin{document}
\maketitle


\begin{abstract}
In this paper, we propose a game-theoretic model to study the properties of the worst-case regret of the greedy strategy under complete (Knightian) uncertainty. In a game between a decision-maker (DM) and an adversarial agent (Nature), Nature chooses an unknown state determining the distribution of ratings for each product. The DM observes a realization of product ratings and then chooses a product according to a strategy. For arbitrary numbers of products and ratings, we first study the equal-observations case in which every product has the same number of observations. In this benchmark, we establish matching upper and lower bounds on the worst-case regret, showing that the regret vanishes as the number of observations increases and that the greedy strategy is rate-optimal up to universal constants. In the special case with two products and two ratings, we show that with one observation per product the greedy strategy is minimax-optimal with respect to worst-case regret. We then allow products to have different numbers of observations. Greedy remains robust in a conservative sense: its worst-case regret is controlled by the least-reviewed product. However, unequal numbers of observations can also change greedy's exact worst-case behavior. In particular, adding observations for only one product can increase greedy's worst-case regret. Finally, we test the model on data collected from Google reviews for restaurants, showing that the greedy strategy's empirical performance closely aligns with the theoretical findings.
\end{abstract}

\setcounter{tocdepth}{1}
\compacttoc
\newpage

\section{Introduction}\label{section:intro_section}

Suppose you select a restaurant based on its average observed rating. How big is the regret if you simply pick the restaurant with the highest average rating?
We introduce a game-theoretic model to study the worst-case regret of greedy selection under Knightian uncertainty. In a game against Nature, the decision-maker (DM) observes a realization of a stochastic ratings matrix (i.e., sampled product ratings for each product). Upon observation, the DM chooses a product according to a strategy (pure or mixed), which maps each observation to a probability distribution over the set of products.
Nature chooses a state, that is, a stochastic matrix of probability distributions over the product ratings, which is unobservable to the DM. 
We allow Nature to choose any admissible state and the DM to use any possibly randomized strategy.

The basic structure of our framework dates back to Abraham Wald's maximin model \cite{Wald1950} and Leonard Savage's minimax regret model \cite{Savage1951}, both of which are based on worst-case principles. 
Wald's maximin criterion prescribes choosing the strategy that maximizes the minimum payoff, while Savage's criterion prescribes choosing the strategy that minimizes the maximum regret. The connection to our setting is that the DM tries to minimize the regret against the worst-case probability distribution.

In this paper, we study arguably the most natural strategy the DM can employ: the \textit{greedy} strategy, where the DM picks the product with the highest average rating and breaks ties uniformly at random. We measure how the greedy strategy performs in different scenarios, focusing on the worst-case. We track two metrics, the payoff that the DM receives from the game and the regret, which is defined as the difference between the expected payoff of the best possible alternative and the DM's product choice, given the true state.

\subsection{Results Overview}

We present several results on the regret of the greedy strategy, including both theoretical bounds and empirical application using Google restaurant reviews.

First, for the case with equal number of observations, for any number of products and any number of ratings, we provide an upper bound on the greedy strategy's worst-case regret and show that it vanishes as the number of observations grows (Proposition~\ref{thm:greedy-upper}). We complement this with a lower bound on worst-case regret, showing that no strategy can achieve a fundamentally faster decay (Proposition~\ref{lem:minimax-lower}). Together, these results imply that the greedy strategy achieves the minimax regret rate up to universal constants (Theorem~\ref{cor:minimax_optimal_greedy}).

Second, we consider a special case with two products and two ratings and prove that the greedy strategy incurs at most $\frac{1}{8}$ expected regret after a single observation per product, regardless of the distribution (Proposition ~\ref{proposition:regretm1}). We further show that, with respect to the worst-case regret, the greedy strategy is minimax-optimal and guarantees the lowest possible regret  (Corollary~\ref{cor:greed_optimal}). 

Third, we numerically show in Section~\ref{section:regret_m12} that the regret of the greedy strategy decreases quickly as the number of observations increases from 2 to 20 for each product.

Fourth, we relax the equal-observation assumption and allow the number of observations to differ across products. We show that the finite-sample regret upper bound for greedy continues to hold, constrained by the number of observations of the least observed product (Proposition~\ref{prop:het_up_bound}). \newline
However, the heterogeneous observations case also exhibits a nonmonotonicity that is absent from the equal-observations model. In the two-product, two-rating case, we show that adding observations for only one product can increase greedy's worst-case regret (Proposition~\ref{prop:ones_regret}).

Finally, using a dataset of 1.5 million Google restaurant reviews \cite{YHLZM2022}, we benchmark the greedy strategy on sampled ratings. The empirical experiment is an equal-sampling experiment: every restaurant is sampled with the same number of reviews. The greedy strategy attains lower mean regret than both the uniform strategy and the Thompson Sampling algorithm for every tested pair of (number of products, number of observations). Its regret also declines as the number of observations increases, in line with the equal observations numerical illustration in Section~\ref{section:regret_m12}.

\subsection{Related Work}
\noindent Under complete uncertainty about the underlying distribution, the greedy strategy makes the locally optimal choice for each observation. Greedy algorithms \cite{CLRS1990,Wang2023} are ubiquitous in optimization and learning, but their performance is highly model-dependent. In our setting, greedy selects the product with the highest empirical value from a fixed batch of discrete ratings, and we analyze its worst-case regret under complete uncertainty about the underlying rating distributions.

\noindent The greedy strategy in this paper corresponds to the Randomized Greedy Algorithm considered by Blum and Mansour in 2007 \cite{BM2007}, where they break ties between equal best actions uniformly at random. Their framework evaluates the strategies in a repeated game and measures regret as a function of time (the number of steps taken in the repeated game), whereas our setting is a one-shot selection problem with a batch of observations. Rather than cumulative regret, our objective is measuring worst-case expected regret after observing a fixed batch of ratings. Within the equal-observations model, where every product has the same number of observations, we derive matching upper and lower bounds on greedy's worst-case expected regret.

\noindent The greedy strategy that we are analysing in this paper also closely relates to research by Mahdian et al. \cite{MJK2022} where they measure the regret of the greedy strategy when used to make decisions from noisy observations. Their findings show that the regret incurred by the greedy algorithm can surpass the optimal regret by an unbounded margin, even in cases where the noises are unbiased; also applying to the two product case studied in our paper. Our results complement this perspective, such that product values are inferred from discrete past ratings, and we give a rate-tight worst-case characterization. In particular, in the equal-observations benchmark, we show that the worst-case expected regret vanishes as the common number of observations increases and prove a matching lower bound showing that this decay rate cannot be fundamentally improved by any other strategy.

\noindent Our motivation follows Savage's minimax regret criterion \cite{Savage1951}, developed to minimize worst-case regret when the decision maker is unwilling to commit to a probabilistic prior. Similar in spirit to Wald's maximin model \cite{Wald1950}, it has often been used to formalize choice under uncertainty by evaluating performance against an adversarially chosen state of Nature. While these methods also do not require knowledge of the underlying Nature distribution, they require a set of possible outcomes and then optimize against the worst case within that set.\footnote{Relatedly, Ismail \cite{ismail2025} proposes the optimin criterion, another robustness concept that coincides with Wald's criterion in zero-sum games. See also Renou and Schlag \cite{renou2010} who propose an equilibrium notion based on minimax regret.}

\noindent A closely related line of work studies treatment choice under minimax regret. Manski~\cite{manski2004statistical} analyzes statistical treatment rules with experimental data and derives finite-sample upper bounds for empirical-success rules and their covariate-conditioned variants. Kitagawa and Tetenov~\cite{kitagawa2018who} extend this perspective via empirical welfare maximization over a feasible policy class. Our setting can be viewed as the no-covariate, finite-action counterpart of both papers: when the feasible policies are constant actions and there is no covariate-dependent assignment, the relevant empirical rule is to choose the action with the highest sample-average outcome. This specialization also permits a sharper worst-case characterization than is available in that broader treatment-choice framework: in the equal-observation one-shot finite-action setting, we derive matching upper and lower worst-case regret bounds, so that the rate is pinned down up to constants, and we show that greedy is minimax-rate optimal.

\noindent For equal number of observations across products, our results are related to a literature on exact finite-sample minimax treatment rules. Schlag and Stoye analyze two-treatment problems \cite{schlag2006eleven,stoye2009minimax}, while Chen and Guggenberger show that, under balanced assignment and binary outcomes, uniform selection among the treatments with the most observed successes is minimax-regret optimal for any finite number of treatments \citep{chen2026minimax}. Our two-product, two-rating benchmark establishes an analogous minimax conclusion in our product selection model. More broadly, our contribution in the equal-sample setting is to characterize the distribution-free minimax regret rate of greedy empirical-average selection for an arbitrary number of ratings and products. We then study the greedy strategy under deterministic unequal numbers of observations.

\noindent Several connections also exist to the pure-exploration and bandit literature, which studies recommending a single best action after collecting data. Though greedy rules are often suboptimal in sequential settings \cite{BGY2004,Wang2023}, Bayati et al. \cite{BHJK} and Jedor et al. \cite{JLP2021} show that greedy strategies can perform well in some multi-armed bandit models. Relatedly, Kannan et al.~\cite{kannan2018smoothed} give smoothed-analysis guarantees for greedy in linear contextual bandits, while our setting studies one-shot minimax regret under adversarial uncertainty over discrete rating distributions. Yun et al. \cite{YPASY2018} extend the bandit model to allow observing additional arms at some cost. In our case, observations have no cost but are also drawn from an unknown distribution.

\noindent As a useful benchmark from pure-exploration bandits, Bubeck, Munos, and Stoltz \cite{BMS2011} analyze one-shot regret for selecting the empirically best arm when all arms have the same number of observations. In our setting, each product comes with exactly $m$ i.i.d. ratings, and greedy selects the product with the highest observed value (with uniform tie-breaking). Our upper bound is of the same order as their benchmark. Moreover, in our setting we establish a matching lower bound, showing that no strategy can achieve a fundamentally faster worst-case decay of its regret.

\noindent Regret minimization is commonly studied in no-regret learning algorithms, as discussed by Roughgarden \cite{TN16} and Avramopoulos et al. \cite{ARS2008}. However, similar to the work by Blum and Mansour, it is measured in repeated settings with respect to the number of rounds. Lastly, using a method involving convex-combination updating, Flores-Szwagrzak \cite{Flores2022} shows how an agent can utilize observations from an unknown distribution, to consistently learn and improve their payoff.

\section{The Model}\label{sect:model}
Let $G=(\Sigma_1, \mathcal{S}, \pi)$ be a two-player zero-sum game against Nature, with player one called the Decision Maker (DM), and player two called Nature. In what follows, we define the notation and the terminology used in this paper. 

\noindent Let $n_d \text{ and }n_r\in \mathbb{N}$.
We use
$D = \{1,2,...,n_d\}$ to denote the set of \textit{products}  and $R=\{1,2,...,n_r\}$ to denote the set of discrete \textit{ratings}, where we assume that  $n_r \geq 2$. Nature decides the ratings for each product $d$ with a probability distribution unknown to the DM. The probability distributions for each product and rating are defined in a state (stochastic matrix) $S$. Similarly, an observation matrix $B$ reflects the samples from $S$. The decision-maker picks a strategy $\sigma_1 \in \Sigma_1$ that selects a product based on the observed ratings in $B$ which are sampled from an unknown stochastic-matrix $S \in \mathcal{S}$, with the aim to maximize its expected rating. We use $\pi$ as the payoff function that determines the payoff in the game.

\subsection{States} \label{section:model_states}
Since each product $d$ has a probability distribution for its ratings, the set of products $D$ naturally induces a column-stochastic matrix $S \in [0,1]^{n_r \times n_d}$ called a \textit{state}:
\[
S = 
\begin{bmatrix}
s_{1,1} & s_{1,2} &\cdots & s_{1, n_d}\\
\vdots & \vdots & \ddots & \vdots \\
s_{n_r, 1} & s_{n_r, 2} & \cdots &  s_{n_r, n_d}
\end{bmatrix}. 
\]
Each element $s_{r,d}$ of a state gives the probability that a product $d \in D$ receives a rating $r \in R$, hence, for all $d \in D$, $\sum_{r \in R} s_{r,d}=1$. Let $\mathcal{S}=\{S \in [0,1]^{n_r \times n_d} ~|~ \text{where for all } d \in D, \sum_{r \in R}s_{r,d}=1\}$ be the set of all possible states. \newline The following is an example of a state with two products as the matrix columns and two ratings as its rows:
\[
\tilde{S}_1=
\begin{bmatrix} 
0.7 & 0.4\\
0.3 & 0.6
\end{bmatrix}.
\]
In this state, there is a 0.7 probability that product 1 receives a rating of 1 and a 0.3 probability that it receives a rating of 2. In addition, product 2 has a 0.4 probability of receiving a rating of 1 and a 0.6 probability of receiving a rating of 2.

\subsection{Observations}\label{sec:observations}
\noindent For each product, the DM observes how often each rating occurs, but not the underlying state chosen by Nature. Sections~\ref{sect:model}--\ref{sect:results} focus on the equal-observations case in which every product has the same number of observations $m$. Section~\ref{sec:het_obs_res} then relaxes this assumption and allows product-specific observation counts.
\noindent Let $B \in \mathbb{N}_0^{n_r \times n_d}$ denote an \textit{observation matrix}, where $b_{r,d}$ counts how often product $d$ receives rating $r$:
\[
B = 
\begin{bmatrix}
b_{1,1} & b_{1,2} &\cdots & b_{1, n_d}\\
\vdots & \vdots & \ddots & \vdots \\
b_{n_r, 1} & b_{n_r, 2} & \cdots &  b_{n_r, n_d}
\end{bmatrix}. 
\]
Conditional on the state $S$, the individual rating draws are independent across products and observations. We consider the set $\mathcal{B}_m$ of possible observations, defined by
\[
\mathcal{B}_m =
\{ B \in \mathbb{N}_0^{n_r \times n_d} \; | \text{ where for all } d\in D, \sum_{r\in R} b_{r,d} = m   \},
\]
where $m \in \mathbb{N}$ denotes the total number of observations across all ratings for any product. For example when $m=1$, then there are four possible observation matrices $\mathcal{B}_1=\{\tilde{B}_1, \tilde{B}_2, \tilde{B}_3, \tilde{B}_4\}$:
\begin{align*}
\label{eq:B}
\tilde{B}_1 = \begin{bmatrix} 
1 & 1\\
0 & 0
\end{bmatrix}\text{, }
\tilde{B}_2 = \begin{bmatrix} 
1 & 0\\
0 & 1
\end{bmatrix} \text{, }
\tilde{B}_3 = \begin{bmatrix} 
0 & 1\\
1 & 0
\end{bmatrix} \text{and }
\tilde{B}_4 = \begin{bmatrix} 
0 & 0\\
1 & 1
\end{bmatrix}.
\end{align*}
\newline
Denote by $b_{\cdot,d}=(b_{r,d})_{r\in R}$ the observation vector corresponding to product $d$. For later use, let
\[
\mathbf m=(m_d)_{d\in D}
\]
denote the observation count vector, where
\[ m_d=\sum_{r\in R} b_{r,d} \qquad\text{for all } d\in D. \]
When $m_d=m$ for all products $d$, we simply write $m$.
For a given product $d$ and state $S$, its rating $r$ is drawn from a multinomial distribution characterized by a number of trials $m$ and a vector of probabilities $p=(p_{r_1}, p_{r_2},...,p_{n_r})$, where $p_i = s_{r_i, d}$. Let $Z$ denote a random variable which follows a multinomial distribution with the parameters $m$ and $p$. Observe that with $Z$, the probability of an observation vector per product $d$, $b_{\cdot, d}$ occurring given a state $S$ is:
\begin{align*}
\Pr{b_{\cdot,d} | S}=P(Z_1=b_{1, d}, Z_2=b_{2,d}, ...,Z_{n_r} = b_{n_r, d}) = \frac{m!}{\prod_{r}b_{r,d}!}\prod_{r}S_{r,d}^{b_{r,d}}\in [0,1].
\end{align*} \newline
By the independence of these distributions, the probability of observing matrix $B$ given $S$ is

\begin{align*}
\Pr{B|S} =  \prod_{d \in D}\Pr{b_{\cdot,d} | S} =\frac{(m!)^{n_d}}{\prod_{r,d}b_{r,d}!}\prod_{r,d} S_{r,d}^{b_{r,d}} \in [0,1].
\end{align*} The first term consists of the multinomial coefficient, and the second term is the probability that a given product and rating combination is observed $b_{r,d}$ times.

\subsubsection*{Illustrative example}\label{sec:illustrative_ex}

\noindent Consider the following observation matrix for 2 products and 2 ratings:
\[ 
\tilde{B} = \begin{bmatrix} 
1 & 0\\
2 & 3
\end{bmatrix}.
\]
This matrix shows that, for product 1, the DM observed rating 1 once and rating 2 twice. For product 2, the DM observed rating 2 three times.

\noindent To illustrate our framework intuitively, consider a situation where the DM makes three `draws' (observations) from two urns (products), each containing balls marked 1 or 2 (ratings). These draws follow the probability distributions in a given state. For the example of $\tilde{S}_1$, we would obtain the number of observations given in matrix $\tilde{B}$ when the DM draws, with replacement, one ball marked 1 with a probability of 0.7 and two balls marked 2, each with a probability of 0.3, from urn 1. From urn 2, she draws three balls marked 2, each with a probability of 0.6.

\noindent For product 1 in our illustrative example,
$
\Pr{\tilde{B}_{\cdot,1} | \tilde{S}_1} = \binom{3}{1} \cdot  0.7 \cdot (0.3)^2 = 3 \cdot 0.063 = 0.189.
$
This is because the probability of `drawing' a rating of one is 0.7 and 2 ratings of two is $(0.3)^2$. We multiply this by the Binomial coefficient $\binom{3}{1}$ because there are three different combinations in which one can draw 1 rating of one and 2 ratings of two. Similarly
$
\Pr{\tilde{B}_{\cdot,2} | \tilde{S}_{1}} = \binom{3}{0}  \cdot 0.216 = 0.216.
$
Therefore, the probability of observing $\tilde{B}$ in state $\tilde{S}_1$ is
$
\Pr{\tilde{B}|\tilde{S}_{1}} = \prod_{d \in \{1,2\}}\Pr{\tilde{B}_{\cdot,d} | \tilde{S}_{1}} = 0.040824.
$\newline 

Table~\ref{tab:prob_obs} presents the probability of each observation matrix for $m=1$ occurring for $\tilde{S}_1$.

\begin{table}[H]
\centering
\begin{tabular}{l|l|l|l|l}
 &    $\tilde{B}_1$ & $\tilde{B}_2$ & $\tilde{B}_3$ & $\tilde{B}_4$                       \\ \hline
$\Pr{B|\tilde{S}_1}$      & 0.28  & 0.42 & 0.12& 0.18  \\ 
\end{tabular}
\caption{Probability of each observation $B \in \mathcal{B}_1$ occurring with $\tilde{S}_1$ }
\label{tab:prob_obs}
\end{table}

\subsection{Strategies}
A \textit{strategy} for the DM is a function $\sigma_1:\mathcal{B}_m \rightarrow \mathcal{D}$ where $\mathcal{D}=\{p:D\rightarrow [0,1] ~|~ \sum_{d \in D}p(d)=1\}$ which assigns a probability distribution over the set of products $D$ for every observation matrix $B \in \mathcal{B}_m$. For a given observation matrix $B$, let $\Sigma_1$ be the set of all probability distributions over $D$:
\[\Sigma_1=\left\{\sigma_1(B) \in [0,1]^{n_d} \mid \sum_{d \in D}\sigma_1(B)(d)=1\right\}.\] Nature's strategy is to pick any state $S \in \mathcal{S}$, such that it maximizes the DM's regret, which we define below.

\subsection{Payoff}
 
The expected value, $V_S:D \rightarrow \mathbb{R}$, of a product given state $S$ is defined as follows
\[
V_S(d) = \sum_{r \in R} r\cdot s_{r,d}.
\]
For the $\tilde{S}_1$ state, the values of the products are the following, $V_{\tilde{S}_1}(1)=1.3$ and $V_{\tilde{S}_1}(2)=1.6$; where we used the first and second column of $\tilde{S}_1$ to compute $V_{\tilde{S}_1}(1)$ and $V_{\tilde{S}_1}(2)$ respectively.

\noindent Let $\pi : \Sigma_1 \times \mathcal{S}\rightarrow \mathbb{R}$ denote the DM's \textit{expected payoff function}. For a given strategy $\sigma_1$ and a state $S$, the payoff for the DM is:
\[\pi(\sigma_1,S) = \sum_{B \in \mathcal{B}_m}\Pr{B|S}\sum_{d \in D}\sigma_1(B)(d)V_S(d).\]
\subsection{Regret}
We now introduce the notation for \textit{regret}. For a given state of Nature $S$, the (expected) regret for the DM not selecting the highest value product in that state is
\[
\bar{\gamma}(\sigma_1, S) = \max_{d \in D} V_S(d) - \pi(\sigma_1, S).
\] Since the regret is measured only for the DM, it is non-negative.

\begin{Definition}
For a strategy $\sigma_1$, we define its worst-case regret, $\gamma(\sigma_1)$, as follows:
\[
\gamma(\sigma_1) = \max_{S \in \mathcal{S}} \left[ \max_{d \in D} V_S(d) - \pi(\sigma_1, S) \right].
\]
\end{Definition}

\noindent In words, the regret of strategy $\sigma_1$ is the difference in expected payoff between $\sigma_1$ and the highest valued product at some state $S \in \mathcal{S}$ such that this difference is maximized.

\section{Strategies for the Decision Maker}
\label{section:greedy_strategy}
In this section we define strategies for the DM that select products based on the available observations.

\noindent The \textit{uniform strategy}, denoted as $\sigma_u$, selects each product $d \in D$ with equal probability, $\sigma_u(B) \sim \text{Uniform}(D)$, regardless of the observation matrix $B$. \newline
\noindent The \textit{greedy strategy} selects the product(s) with the highest observed average rating. Given observation matrix $B$, let $ V_B : D \to [1, n_r]$  denote the function computing a product's average observed rating:

\[
V_B(d) = \frac{1}{m} \sum_{r \in R} r \cdot b_{r,d}.
\]

\noindent The greedy strategy then selects uniformly from the set of maximizers:

\[
\sigma_g(B) \sim \text{Uniform} \left( \arg\max_{d \in D} V_B(d) \right).
\]

\noindent In words, if there is a unique $d^*$ with the highest average rating, then the greedy strategy chooses $d^*$ with probability 1. If there are multiple maximizers, then the greedy strategy is the uniform distribution over the set of maximizers in $\arg \max_{d\in D} V_B(d)$. 
Table~\ref{tab:greedy_ex} presents a simple example of a greedy strategy choice of products.
\begin{table}[H]
\centering
\begin{tabular}{l|l|l}
$\sigma_g$ & $1$& $2$ \\ \hline
$\tilde{B}_1$ & 0.5 & 0.5  \\ 
$\tilde{B}_2$ & 0 & 1  \\ 
$\tilde{B}_3$ & 1 & 0  \\ 
$\tilde{B}_4$ & 0.5 & 0.5  \\ 
\end{tabular}
\quad \quad \quad

\caption{The greedy strategy $\sigma_g$ for observations in $\mathcal{B}_1$}
\label{tab:greedy_ex}
\end{table}

\noindent Given $\tilde{B}_1$  or $\tilde{B}_4$ defined in Section~\ref{sec:observations}, both products have identical observed ratings, so the greedy strategy selects uniformly between them. For $\tilde{B}_2$, product 2 has a higher rating and is chosen with probability 1; similarly, for $\tilde{B}_3$, product 1 is selected with probability 1.
\newline
For $\tilde{B}_1$ and state $\tilde{S}_1$, the greedy strategy plays both products uniformly. \newline Its expected payoff is $\Pr{\tilde{B}_1|\tilde{S}_1}\left(\sigma_g(\tilde{B}_1)(1)V_{\tilde{S}_1}(1) +\sigma_g(\tilde{B}_1)(2)V_{\tilde{S}_1}(2)\right)=0.28\left(\frac{1.3}{2} + \frac{1.6}{2}\right)=0.406$.
Averaging over all $B \in \mathcal{B}_1$, we compute the expected payoff of greedy: $\pi(\sigma_g, \tilde{S}_1) = 1.495$. The regret is the value gap between the best product and this expected payoff: $V_{\tilde{S}_1}(2)-\pi(\sigma_g, \tilde{S}_1)=1.6-1.495=0.105$.

\subsection{Upper Confidence Bound Algorithm}
The UCB algorithm balances exploration and exploitation by selecting the product with the highest upper confidence bound on its value \cite{ACF2002}. For a generic observation vector $\mathbf m=(m_d)_{d\in D}$, fix a confidence term function $\beta:\mathbb N\to\mathbb R_+$ that is decreasing. Given an observation matrix $B$, define the UCB score of product $d$ by
\[
I_d(B)=V_B(d)+\beta(m_d).
\]
The UCB strategy then is
\[
\sigma_{\mathrm{UCB}}(B)
\sim
\operatorname{Uniform}
\left(
\arg\max_{d\in D} I_d(B)
\right).
\]
In the equal-observations case, $m_d=m$ for every product $d$, so the confidence term $\beta(m_d)$ is constant across products. Therefore UCB selects the product with the highest observed value and is identical to the greedy strategy. Therefore, we do not analyze it separately in Section~\ref{sect:results}. When the number of observations per product is different, this equivalence no longer holds; Section~\ref{sect:UCB_het} returns to this point with a short counterexample.

\subsection{Thompson Sampling}

Thompson Sampling is a Bayesian algorithm that balances exploration and exploitation commonly used in multi-armed bandits \cite{RRKO2017}. In our one-shot setting, we use it only as an empirical benchmark after the observation matrix has already been realized. Fix a prior parameter $\alpha^0\in(0,\infty)^{n_r}$. Given the observed count vector $b_{\cdot,d}$ for product $d$, define $\alpha(d)=\alpha^0+b_{\cdot,d}$.
For each product $d\in D$, Thompson Sampling draws once independently from the Dirichlet distribution $Y_d\sim \operatorname{Dir}(\alpha(d))$, and then selects uniformly from the products with the highest sampled value:\footnote{Beta distribution if there are only two ratings.}
\[
\sigma_{\mathrm{TS}}(B)
\sim
\operatorname{Uniform}
\left(
\arg\max_{d\in D}
\sum_{r\in R} rY_{d,r}
\right).
\]

For example, with two ratings and the observation matrix
\[
\widetilde B_5=
\begin{bmatrix} 
7 & 5\\
2 & 4
\end{bmatrix},
\]
using the uniform prior $\alpha^0=(1,1)$ gives posterior draws
$Y_1\sim\operatorname{Beta}(3,8)$ and
$Y_2\sim\operatorname{Beta}(5,6)$ over the probability of the higher rating. The rule then selects the product with the larger sampled value.

We do not use sequential-bandit regret guarantees for Thompson Sampling in the theoretical analysis below. Those guarantees concern observations collection and cumulative regret over time, whereas our model is a one-shot decision problem after an exogenously observed batch of ratings. Thompson Sampling is therefore included as a comparison rule in the empirical section rather than as a theoretically analyzed benchmark.

\section{Results: Equal Number Of Observations} \label{sect:results}
This section states our main regret bounds for the greedy strategy in our model. There are $n_d$ products, ratings take values in $\{1,\dots,n_r\}$, and we observe a fixed number $m$ i.i.d.\ ratings per product.
We begin with the main minimax statement, then present the two propositions that imply it; full proofs are given later in Sections \ref{sec:proof-lb} and \ref{sec:proof-ub}. We additionally handle the special case for $n_d=2$ in Section \ref{section:regret_m1}: in this special case we can go beyond rate
statements and prove a proposition showing that greedy is exactly minimax-optimal (it attains the minimax regret, not just the minimax rate). In all results that follow, all logarithms are natural. \newline
Recall that for a strategy $\sigma_1 \in \Sigma_1$ and state $S \in \mathcal{S}$, the (expected) regret is the difference of the payoffs of the best product (in hindsight) and the product selected by a particular strategy $\sigma_1$
\[
\bar\gamma(\sigma_1,S) = \max_{d\in D}V_S(d) - \pi(\sigma_1,S).
\]
We measure worst-case performance by maximizing over states:
\[
\gamma(\sigma_1) = \max_{S\in\mathcal S}\bar\gamma(\sigma_1,S).
\]
\begin{theorem}[Greedy is minimax-rate optimal]\label{cor:minimax_optimal_greedy}
Assuming $m \ge \frac{3}{64}\log(n_d)$, there exist universal constants\label{cor:minimax_optimal_greedy} $0<c_0\le C_0<\infty$ such that both of the following hold \footnote{The assumption of $m \ge \frac{3}{64}\log(n_d)$, is largely immaterial, since it is trivially fulfilled for $n_d\approx10^9$ when $m=1$.}
 \begin{itemize}
\item No strategy can have regret better than \[ c_0(n_r-1)\sqrt{\frac{\log(n_d)}{m}}  \] 
\item The greedy strategy has a regret of at most 
\[C_0(n_r-1)\sqrt{\frac{\log(n_d)}{m}}.\]
\end{itemize}

In particular, $\gamma(\sigma_g)=\Theta\!\left((n_r-1)\sqrt{\frac{\log(n_d)}{m}}\right)$ and greedy is minimax-rate optimal in this regime.
\end{theorem}
\noindent We call a strategy \textit{minimax-rate optimal} if its worst-case regret matches, up to constant factors, the smallest possible worst-case regret over all strategies $\sigma_1 \in \Sigma_1$, with $n_d$ and $n_r$ fixed.

\noindent Theorem \ref{cor:minimax_optimal_greedy} follows by combining the regret lower bound (proved for any strategy, not just greedy) with a matching upper bound for the greedy strategy.

\begin{proposition}[Regret lower bound]\label{lem:minimax-lower}
For $n_d \ge 3$ and $m \ge \frac{3}{64}\log(n_d)$, there exists a universal constant $c>0$ such that,
\[
\inf_{\sigma_1 \in \Sigma_1} \gamma(\sigma_1) \ge c (n_r-1) \sqrt{\frac{\log(n_d)}{m}}.
\]
\end{proposition}
\noindent\textbf{Proof idea:} We construct a hard family of states in which all products appear nearly identical: exactly one product is only slightly more likely to receive the top rating, while all others are slightly less likely. When the advantage is on the order of $\sqrt{\frac{\log(n_d)}{m}}$, the observation matrix does not contain enough information to reliably identify the best product, so any strategy misidentifies it with nontrivial probability. The proof in Section \ref{sec:proof-lb} bounds pairwise KL divergences and then applies Fano's inequality.

\begin{proposition}[Finite-sample worst-case regret of greedy] \label{thm:greedy-upper}
For any number of products $n_d \ge 1$, ratings $n_r \ge 2$, and observations $m \ge 1$, the greedy strategy satisfies
\[ \gamma(\sigma_g) = \max_{S \in \mathcal S} \bar\gamma(\sigma_g, S) \;\le\; (n_r - 1)\sqrt{\frac{2\log(n_d)}{m}}.\]
In particular,
\[ \gamma(\sigma_g) = O\!\left((n_r - 1)\sqrt{\frac{\log(n_d)}{m}}\right) \quad\text{as } m \to \infty.\]
\end{proposition}
\noindent\textbf{Proof idea.} Greedy can incur regret only if the best product is underestimated or the suboptimal product(s) is overestimated. Because ratings are bounded, such upward and downward fluctuations of empirical averages are known, and the worst-case regret is therefore governed by how large these fluctuations can be when we compare many products using only $m$ observations per product. The proof in Section \ref{sec:proof-ub} uses Hoeffding's inequality, to bound the products' empirical averages.

\subsection{Greedy's (upper bound) regret and asymptotic behavior} \label{sec:proof-ub}
In this subsection we provide an upper bound on the worst-case regret of the greedy strategy. The bound scales with the number of products and observations as $(n_r - 1)\sqrt{\frac{\log(n_d)}{m}}$ and holds over all states of Nature. In particular, it does not require any assumption on the gap between the best and the second-best product in the state. We start by recording a concentration inequality for the observed (empirical) product values.

\begin{lemma}[Concentration of observed values] \label{lem:concentration-VB}
Fix a state $S \in \mathcal S$ and a product $d \in D$. For each $k \in \{1,\dots,m\}$ let $X_{k,d} \in R$ denote the $k$-th rating observed for product $d$, so that
\[ V_B(d) = \frac{1}{m}\sum_{k=1}^m X_{k,d} \quad\text{and}\quad V_S(d) = \mathbb{E}[X_{k,d}\mid S]. \]
Then, for every $t \ge 0$,
\[ \Pr{|V_B(d) - V_S(d)| \; \ge t \;|\;S}  \le 2 \exp\left(-\frac{2m t^2}{(n_r - 1)^2}\right).\]
\end{lemma}

\begin{proof}
Each random variable $X_{k,d}$ takes values in the set $\{1,\dots,n_r\}$, so the difference between its maximum and minimum possible value is $n_r - 1$. The $m$ observations for a fixed product $d$ are independent by construction. The claim therefore follows directly from Hoeffding's inequality for averages of independent bounded random variables applied to the sequence $X_{1,d},\dots,X_{m,d}$.
\end{proof}
Lemma~\ref{lem:concentration-VB} implies that, for a fixed state $S$ and product $d$, the estimation error $V_B(d) - V_S(d)$ behaves like a sub-Gaussian random variable with parameter $\frac{(n_r - 1)^2}{ 4m}$. Since the columns of the observation matrix $B$ are independent given  $S$, these errors are independent across products.
Therefore, by the standard bound on the expected maximum of independent sub-Gaussian random variables,
\begin{align}
\mathbb{E} \left[\max_{d\in D}\left(V_B(d)-V_S(d)\right) \mid S\right]
&\le (n_r-1)\sqrt{\frac{\log(n_d)}{2m}}, \label{eq:max-over}\\
\mathbb{E} \left[\max_{d\in D}\left(V_S(d)-V_B(d)\right) \mid S\right]
&\le (n_r-1)\sqrt{\frac{\log(n_d)}{2m}}. \label{eq:max-under}
\end{align}

We can now (upper) bound the regret of the greedy strategy, proving Proposition \ref{thm:greedy-upper}.

\begin{proof}
Fix a state $S \in \mathcal S$ and let $d^* \in \arg\max_{d \in D} V_S(d)$ be the best valued product in state $S$. Given an observation matrix $B \in \mathcal{B}_m$, let $\hat{d} = \sigma_g(B)$ be the product selected by the greedy strategy. For this fixed $B$ and $S$, the instantaneous regret is $V_S(d^*) - V_S(\hat{d})$.

Since greedy selects the maximum value observed product, its regret can be decomposed in two parts, the amount by which the true best product's value $V_S(d^*)$ is underestimated in the given observation matrix: $V_S(d^*)-V_B(d^*)$ and the amount by which the greedy selected product's value $V_B(\hat{d})$ is overestimated in the given observation matrix: $V_B(\hat{d})-V_S(\hat{d})$:
\[V_S(d^*)-V_S(\hat{d}) \leq (V_S(d^*)-V_B(d^*)) + \left(V_B(\hat{d})-V_S(\hat{d}) \right)\]
Each of the two terms on the right-hand side can be bounded by a maximum over the products:
\[ V_S(d^*) - V_B(d^*) \leq \max_{d \in D} (V_S(d) - V_B(d)) \; \text{ and } \; V_B(\hat{d})-V_S(\hat{d}) \leq \max_{d \in D} (V_B(d) - V_S(d)). \]
Therefore, for every $B$,
\[ V_S(d^*) - V_S(\hat{d}) \leq \max_{d \in D} (V_S(d) - V_B(d)) + \max_{d \in D} (V_B(d) - V_S(d)).\]
Taking expectations with respect to the randomness in $B$ and $\hat d$ gives
\[\bar\gamma(\sigma_g,S) = \mathbb{E}\left[V_S(d^*) - V_S(\hat{d})\right] \leq \mathbb{E}\left[\max_{d \in D} (V_S(d) - V_B(d))\right] + \mathbb{E}\left[\max_{d \in D} (V_B(d) - V_S(d))\right].\]
Applying \eqref{eq:max-over} and \eqref{eq:max-under} to each term yields
\[\bar\gamma(\sigma_g,S) \leq (n_r - 1)\sqrt{\frac{\log(n_d)}{2m}} + (n_r - 1)\sqrt{\frac{\log(n_d)}{2m}} = (n_r - 1)\sqrt{\frac{2\log(n_d)}{m}}. \]
Since the bound does not depend on $S$, taking the maximum over all states completes the proof.
\end{proof}
Proposition~\ref{thm:greedy-upper} shows that the greedy strategy is worst-case robust: its regret vanishes at rate $(n_r - 1)\sqrt{\frac{\log(n_d)}{m}}$ as the number of observations per product grows, for any state of Nature.

\subsection{Regret lower bound}\label{sec:proof-lb}
A clean way to see why regret cannot vanish faster than order $(n_r-1)\sqrt{\frac{\log(n_d)}{m}}$ is to build a family of states that are genuinely hard to distinguish from the observations.
Intuitively, we make all products look almost identical: in each state, exactly one product is slightly more likely to receive a higher rating, while every other product is slightly less likely. If the advantage is small enough, then with only $m$ observations per product, the observation matrix $B$ does not contain enough information to reliably identify which product is truly best.
Any strategy must therefore make mistakes with nontrivial probability, and each mistake costs a fixed value gap, hence nontrivial regret.

\subsection*{A set of hard states.}
Fix $\varepsilon \in \left(0,\frac{1}{2}\right)$. For each $d^*\in D$, define a state $S_{d^*}\in \mathcal{S}$ supported only on ratings $\{1,n_r\}$ by specifying,
\[
s_{n_r,d}=\begin{cases}
\frac{1}{2}+\varepsilon&\text{if }d=d^*,\\
\frac{1}{2}-\varepsilon&\text{if }d\neq d^*,
\end{cases}
\quad
s_{1,d}=1-s_{n_r,d},
\quad
s_{r,d}=0\text{ for all }r\in R\setminus\{1,n_r\}.
\]
Recall $V_S(d)=\sum_{r\in R}r\,s_{r,d}$. In each state $S_{d^*}$, let $d^*$ denote the unique maximizer of $V_{S_{d^*}}(\cdot)$, i.e., $\arg\max_{d\in D}V_{S_{d^*}}(d)=\{d^*\}$. Then, for every $d \in D \setminus\{d^*\}$, the product value gap equals:
\[
V_{S_{d^*}}(d^*)-V_{S_{d^*}}(d)=(n_r-1)\left(\left(\frac{1}{2}+\varepsilon\right)-\left(\frac{1}{2}-\varepsilon\right)\right)=2(n_r-1)\varepsilon.
\]
Fix $d^* \in D$ and any strategy $\sigma_1$. Given $B\in\mathcal{B}_m$, let $\hat d$ be the random product selected according to $\sigma_1(B)$, i.e., $\Pr{\hat d=d\mid B}=\sigma_1(B)(d)\text{ for all }d\in D$.
Then, $\bar\gamma(\sigma_1,S_{d^*})=2(n_r-1)\varepsilon\cdot\Pr{\hat d\neq d^*\mid S_{d^*}}$, or:
\[
\bar\gamma(\sigma_1,S_{d^*})=2(n_r-1)\varepsilon\sum_{B\in\mathcal{B}_m}\sum_{d\in D\setminus\{d^*\}}\Pr{B\mid S_{d^*}}\sigma_1(B)(d).
\]
That is the regret for any strategy of the decision-maker $\sigma_1$ given $S_{d^*}$, the product of the products' value gap and the probability of the strategy misidentifying the best product $d^*$.

\subsection*{KL divergence between observation distributions.}
For $d^*,d\in D$, we write
\[
\mathrm{KL}(S_{d^*}\|S_d)=\sum_{B\in\mathcal{B}_m}\Pr{B|S_{d^*}}\log\left(\frac{\Pr{B|S_{d^*}}}{\Pr{B|S_d}}\right).
\]
We will work with $\varepsilon\in \left (0,\frac{1}{4}\right]$, which is stricter than needed for the construction (which only requires $\varepsilon<\frac{1}{2}$).

\begin{lemma}[Pairwise KL divergence between hard states]\label{lem:kl_pairwise}
Fix $\varepsilon\in(0,\frac{1}{4}]$ and the hard family $\{S_{d^*}\}_{d^*\in D}$ defined above.
For any distinct $d^*,d\in D$,
\[
\mathrm{KL}(S_{d^*}\|S_{d})\le\frac{64}{3}m\varepsilon^2.
\]
\end{lemma}
\begin{proof}
Under $S_{d^*}$, each observed rating lies in $\{1,n_r\}$. For each product $d \in D$ and sample index $k\in\{1,\ldots,m\}$ define
\[
X_{d,k}=\mathbf{1}\{\text{the $k$-th rating of product $d$ equals }n_r\}.
\]
Then the collection $\{X_{d,k}\}_{d,k}$ is independent given the state, and
\[
X_{d,k}\sim\mathrm{Bern}\left(\frac{1}{2}+\varepsilon\right)\text{ if }d=d^*
\text{ and }
X_{d,k}\sim\mathrm{Bern}\left(\frac{1}{2}-\varepsilon\right)\text{ if }d\neq d^*.
\]
Fix distinct $d^*,d\in D$. The states $S_{d^*}$ and $S_d$ differ only in columns $d^*$ and $d$, so by additivity of KL for product measures,
\[
\mathrm{KL}(S_{d^*}\|S_d)=mA+mB,
\]
where
\[
A=\mathrm{KL}\left(\mathrm{Bern}\left(\frac{1}{2}+\varepsilon\right)\Big\| \; \mathrm{Bern}\left(\frac{1}{2}-\varepsilon\right)\right),
\quad
B=\mathrm{KL}\left(\mathrm{Bern}\left(\frac{1}{2}-\varepsilon\right)\Big\| \;\mathrm{Bern}\left(\frac{1}{2}+\varepsilon\right)\right).
\]
A direct simplification gives
\[
A=2\varepsilon\log \left(\frac{1+2\varepsilon}{1-2\varepsilon}\right).
\]
Let $u=2\varepsilon\in \left(0,\frac{1}{2} \right]$. Using $\log \left(\frac{1+u}{1-u}\right) \le \frac{2u}{1-u^2}$, we obtain
\[
A\le 2\varepsilon\cdot\frac{2(2\varepsilon)}{1-(2\varepsilon)^2} =\frac{8\varepsilon^2}{1-4\varepsilon^2} \le\frac{8\varepsilon^2}{\frac{3}{4}} =\frac{32}{3}\varepsilon^2,
\]
where the last inequality uses $\varepsilon\le\frac{1}{4}$. By symmetry the same bound holds for $B$. Therefore,
\[
\mathrm{KL}(S_{d^*}\|S_d) \le m\cdot\frac{32}{3}\varepsilon^2+m\cdot\frac{32}{3}\varepsilon^2=\frac{64}{3}m\varepsilon^2.
\]
\end{proof}
In Lemma \ref{lem:kl_pairwise} we use the Bernoulli distribution for describing the two distinct product draws from the state. For multiple product draws, the probability of an observation matrix comes from a Binomial distribution. In Appendix \ref{section:kl_equals}, Lemma \ref{lem:binom_bern_kl} we show that the KL divergence of the Binomial distribution equals $m$ times the Bernoulli KL divergence.

\subsection*{From KL to misidentification via Fano}
We now translate the KL bound from Lemma \ref{lem:kl_pairwise} into a lower bound on the probability that any strategy misidentifies the best product in the hard family $\{S_{d^*}\}_{d^*\in D}$.

\begin{lemma}[Fano inequality for the hard family]\label{lem:fano_hard_family}
Fix any strategy $\sigma_1$ and let $\hat{d}$ be the (mixed) product selected after observing $B$.
Then
\[
\frac{1}{n_d}\sum_{d^*\in D}\Pr{\hat d \neq d^* \mid S_{d^*}}\ge 1-\frac{\frac{1}{{n_d}^2}\sum_{d^*\in D}\sum_{d\in D}\mathrm{KL}(S_{d^*}\|S_d)+\log(2)}{\log(n_d)}.
\]
\end{lemma}

\begin{proof}
Let Nature draw a state $S_{d^*}$ uniformly, that is $d^* \sim \text{Uniform}(D)$. We draw $B$ according to $\Pr{B \mid S_{d^*}}$, then
\[
\Pr{\hat d\neq d^*}=\frac{1}{n_d}\sum_{d^* \in D}\Pr{\hat d\neq d^* \mid S_{d^*}}.
\]
Fano's inequality gives
\[
\Pr{\hat d\neq d^*}\ge 1-\frac{I(d^*;B)+\log(2)}{\log(n_d)}.
\]
It remains to upper bound $I(d^*;B)$. Define the uniform mixture distribution on $\mathcal{B}_m$ by
\[
\bar P[B] =\frac{1}{n_d}\sum_{d^*\in D} \Pr{B \mid S_{d^*}}.
\]
By the standard identity for mutual information under a uniform prior,
\[
I(d^*;B) =\sum_{d^*\in D}\frac{1}{n_d}\sum_{B'\in \mathcal{B}_m} \Pr{B'\mid S_{d^*}}\, \log \left(\frac{\Pr{B'\mid S_{d^*}}}{\bar P[B']}\right) =\frac{1}{n_d}\sum_{d^*\in D}\mathrm{KL} \left(\Pr{\cdot\mid S_{d^*}}\;\|\;\bar{P}\right).
\]
Fix $d^* \in D$. Using the definition of $\bar{P}$ and the convexity of $-\log(\cdot)$, for each $B\in \mathcal{B}_m$ we have
\begin{align*}
\log\frac{\Pr{B\mid S_{d^*}}}{\bar P[B]} &= -\log \left(\frac{\bar P[B]}{\Pr{B\mid S_{d^*}}}\right)
= -\log \left(\frac{1}{n_d}\sum_{d\in D}\frac{\Pr{B\mid S_d}}{\Pr{B\mid S_{d^*}}}\right)\\
&\le \frac{1}{n_d}\sum_{d\in D} -\log \left(\frac{\Pr{B\mid S_d}}{\Pr{B\mid S_{d^*}}}\right)
= \frac{1}{n_d}\sum_{d\in D}\log \left( \frac{\Pr{B\mid S_{d^*}}}{\Pr{B\mid S_d}}\right).
\end{align*}
Multiplying by $\Pr{B\mid S_{d^*}}$ and summing over $B\in \mathcal{B}_m$ yields
\begin{align*}
\mathrm{KL}\left(\Pr{\cdot\mid S_{d^*}}\;\|\;\bar P\right) =\sum_{B\in \mathcal{B}_m} \Pr{B\mid S_{d^*}}\log \left(\frac{\Pr{B\mid S_{d^*}}}{\bar P[B]}\right) \le \\ \le \frac{1}{n_d}\sum_{d\in D}\sum_{B\in \mathcal{B}_m} \Pr{B\mid S_{d^*}}\log \left(\frac{\Pr{B\mid S_{d^*}}}{\Pr{B\mid S_d}}\right) = \frac{1}{n_d}\sum_{d\in D} \mathrm{KL}(S_{d^*}\|S_d).
\end{align*}
Averaging this bound over $d\in D$ gives
\[
I(d^*;B) = \frac{1}{n_d}\sum_{d^*\in D}\mathrm{KL} \left(\Pr{\cdot\mid S_{d^*}}\;\|\;\bar P\right) \le \frac{1}{{n_d}^2}\sum_{d^*\in D}\sum_{d\in D}\mathrm{KL}(S_{d^*}\|S_d).
\]
Finally, combining this with Fano's inequality,
\[
\Pr{\hat d\neq d^*} \ge 1-\frac{I(d^*;B)+\log(2)}{\log(n_d)},
\]
and using $\Pr{\hat d\neq d^*}=\frac{1}{n_d}\sum_{d^*\in D}\Pr{\hat d\neq d^*\mid S_{d^*}}$, we obtain
\[
\frac{1}{n_d}\sum_{d^*\in D}\Pr{\hat d\neq d^*\mid S_{d^*}}
\ge 1-\frac{\frac{1}{n_d^2}\sum_{d^* \in D}\sum_{d\in D}\mathrm{KL}(S_{d^*}\|S_d)+\log(2)}{\log(n_d)},
\]
which matches the lemma statement.
\end{proof}

\subsection*{Decision-maker's strategy lower bound regret.}
\noindent In this subsection, we complete the proof for Proposition \ref{lem:minimax-lower}, showing the rate of minimum regret any strategy $\sigma_1$ attains.

\begin{proof}
Since $\gamma(\sigma_1)=\max_{S \in \mathcal{S}}\bar{\gamma}(\sigma_1,S)$, we have
\[\gamma(\sigma_1)\geq\frac{1}{n_d}\sum_{d^* \in D}\bar{\gamma}(\sigma_1, S_{d^*}).\]
Fix any strategy $\sigma_1$ and consider the hard family $\{S_{d^*}\}_{d^*\in D}$ defined above with 
\[\varepsilon = \sqrt{\frac{3\log(n_d)}{1024m}}.\] In state $S_{d^*}$, product $d^*$ is the unique maximizer of $V_{S_{d^*}}(\cdot)$ and, for every $d\neq d^*$,
\[
V_{S_{d^*}}(d^*)-V_{S_{d^*}}(d) = 2(n_r-1)\varepsilon.
\]
Recalling that $\hat d$ is the (random) product selected according to $\sigma_1(B)$ after observing $B$, it follows that the expected regret in state $S_{d^*}$ is exactly
\[
\bar\gamma(\sigma_1,S_{d^*})=2(n_r-1)\varepsilon \cdot \Pr{\hat d\neq d^* \mid S_{d^*}}.
\]
Therefore,
\begin{align*}
\gamma(\sigma_1) = \max_{S\in S} \bar{\gamma}(\sigma_1,S) &\ge \max_{d^* \in D}\bar{\gamma}(\sigma_1,S_{d^*}) \ge \frac{1}{n_d}\sum_{d^* \in D}\bar\gamma(\sigma_1,S_{d^*}) \\
&= 2(n_r-1)\varepsilon\cdot \frac{1}{n_d}\sum_{d^*\in D}\Pr{\hat d\neq d^* \mid S_{d^*}}.
\end{align*}
Assume $n_d\ge 3$. By Lemma \ref{lem:fano_hard_family} and Lemma \ref{lem:kl_pairwise},
\[
\frac{1}{n_d}\sum_{d^*\in D}\Pr{\hat d\neq d^* \mid S_{d^*}} \ge 1-\frac{\frac{1}{n_
d^2}\sum_{d^*\in D}\sum_{d\in D} \mathrm{KL}(S_{d^*}\|S_d)+\log(2)}{\log(n_d)} \ge 1-\frac{\frac{64}{3}m\varepsilon^2+\log(2)}{\log(n_d)}.
\]
Choose $\varepsilon = \sqrt{\frac{3\log(n_d)}{1024\,m}}$, which lies in $(0,\frac{1}{4}]$ whenever $m\ge \frac{3}{64}\log(n_d)$. Then $\frac{64}{3}m\varepsilon^2=\frac{1}{16}\log(n_d)$ and hence
\[
\frac{1}{n_d}\sum_{d^*\in D}\Pr{\hat d\neq d^* \mid S_{d^*}} \ge \frac{15}{16}-\frac{\log(2)}{\log(n_d)} \ge \frac{1}{4},
\]
where the last inequality holds for all $n_d\ge 3$.
Plugging this back yields
\[
\gamma(\sigma_1)\ge 2(n_r-1)\varepsilon\cdot \frac{1}{4} = \frac{(n_r-1)\varepsilon}{2} = \frac{\sqrt{3}}{64}(n_r-1)\sqrt{\frac{\log(n_d)}{m}}.
\]
Since $\sigma_1$ was arbitrary, we conclude that there exists a universal constant $c>0$ such that
\[
\inf_{\sigma_1}\gamma(\sigma_1)\ \ge\ c\,(n_r-1)\sqrt{\frac{\log(n_d)}{m}} \qquad\text{for all } n_d\ge 3 \text{ and } m\ge \tfrac{3}{64}\log(n_d).
\]
\end{proof}

\subsection*{Regret lower bound for the two product case}
The above Fano-based argument from Lemma \ref{lem:fano_hard_family} requires $n_d\ge 3$ (for $n_d=2$ the bound becomes vacuous, i.e. it collapses to a non-positive regret on the RHS), so we treat the binary case $n_d=2$ separately below.

\begin{lemma}[Regret lower bound for $n_d=2$] \label{lemma:rlb_nd2}
Assume $n_d=2$. There exists a universal constant $c>0$ such that
\[
\inf_{\sigma_1\in\Sigma_1}\gamma(\sigma_1)\ge c (n_r-1)\sqrt{\frac{\log(2)}{m}}.
\]
\end{lemma}

\begin{proof}
Fix $\varepsilon = \sqrt{\frac{3}{512m}}\le \frac{1}{4}$ and consider the two states $S_1,S_2\in\mathcal S$ supported only on ratings $\{1,n_r\}$, given by
\[
S_1=\begin{bmatrix}
\frac12-\varepsilon & \frac12+\varepsilon\\
0 & 0\\[-2pt]
\vdots & \vdots\\[-2pt]
0 & 0\\
\frac12+\varepsilon & \frac12-\varepsilon
\end{bmatrix},\qquad
S_2=\begin{bmatrix}
\frac12+\varepsilon & \frac12-\varepsilon\\
0 & 0\\[-2pt]
\vdots & \vdots\\[-2pt]
0 & 0\\
\frac12-\varepsilon & \frac12+\varepsilon
\end{bmatrix}.
\]
In state $S_i$, the unique maximizer is $d^*(S_1)=1$ and $d^*(S_2)=2$, and the value gap equals $2(n_r-1)\varepsilon$.
Let $\hat d = \sigma_1(B)$ be the (random) product selected by strategy $\sigma_1$ after observing $B$.
Thus for $i\in\{1,2\}$,
\[
\bar\gamma(\sigma_1,S_i)\;=\;2(n_r-1)\varepsilon \mathbb P_{S_i}\big[\hat d\neq d^*(S_i)\big].
\]
Let $P_i(\cdot)= \Pr{\cdot\mid S_i}$ be the law of $B$ under $S_i$, and define the average misidentification probability
\[
P_e(\sigma_1) = \frac12 \mathbb P_{S_1}[\hat d\neq 1]+\frac12 \mathbb P_{S_2}[\hat d\neq 2].
\]
Since $\gamma(\sigma_1)=\max_{S\in\mathcal S}\bar\gamma(\sigma_1,S)\ge \frac12\bar\gamma(\sigma_1,S_1)+\frac12\bar\gamma(\sigma_1,S_2)$,
we get
\[
\gamma(\sigma_1) \ge 2(n_r-1)\varepsilon P_e(\sigma_1).
\]
By the two-point (Le Cam) bound,
\[
\inf_{\sigma_1\in\Sigma_1} P_e(\sigma_1)\;\ge\;\frac12\Big(1-\mathrm{TV}(P_1,P_2)\Big)  \ge \frac12\Big(1-\sqrt{\tfrac12 \mathrm{KL}(P_1\|P_2)}\Big),
\] where $\mathrm{TV}$ is total variation and the last term is Pinsker's inequality.
Moreover, $S_1$ and $S_2$ differ in exactly the same way as two distinct hard states in Lemma \ref{lem:kl_pairwise}, hence
\[
\mathrm{KL}(P_1\|P_2)=\mathrm{KL}(S_1\|S_2) \le \frac{64}{3} m\varepsilon^2
 = \frac{64}{3} m \cdot \frac{3}{512m} = \frac{1}{8}.
\]
Therefore $\inf_{\sigma_1} P_e(\sigma_1)\ge \frac12\left(1-\sqrt{\frac{1}{16}}\right)=\frac{3}{8}$, and so
\[
\inf_{\sigma_1\in\Sigma_1}\gamma(\sigma_1) \ge 2(n_r-1) \varepsilon \cdot \frac{3}{8}
= \frac{3}{4}(n_r-1)\sqrt{\frac{3}{512m}}
= c(n_r-1)\frac{1}{\sqrt m},
\]
for $c=\frac{3}{4}\sqrt{\frac{3}{512}}>0$.
\end{proof}

\subsection{Greedy's regret rate is minimax optimal}\label{sec:minimax_optimal}
In this section we prove Theorem \ref{cor:minimax_optimal_greedy}, showing that the greedy strategy attains the minimax-optimal regret rate (up to universal constants).

\begin{proof}
The upper bound $\gamma(\sigma_g)\le C_0 (n_r-1)\sqrt{\frac{\log(n_d)}{m}}$ follows from Proposition \ref{thm:greedy-upper}.
The lower bound $\inf_{\sigma_1\in\Sigma_1}\gamma(\sigma_1)\ge c_0 (n_r-1)\sqrt{\frac{\log(n_d)}{m}}$ follows from Proposition \ref{lem:minimax-lower} and Lemma \ref{lemma:rlb_nd2} under the stated conditions on $m$. Combining the two completes the proof.
\end{proof}

\subsection{Optimality of the greedy strategy: the two product case} \label{section:regret_m1}
Since we measure the maximum regret across any possible state, without loss of generality we can pick a single generic state, assuming that it maximizes the regret. For $p_1, p_2 \in [0,1]$, define
\[S'=
\begin{bmatrix} 
p_1 & p_2\\
1 - p_1 & 1 - p_2
\end{bmatrix}.
\]
The following proposition shows that for two products, that also have two ratings and with one observation in the game, the worst-case regret for the greedy strategy is $\frac{1}{8}$.
\begin{proposition}\label{proposition:regretm1}
For $n_d=2$, $n_r=2$ and $m=1$, the worst-case regret of the greedy strategy is $\frac{1}{8}$.
\end{proposition}
\begin{proof} 
Given $S'$ the value of each product is: $V_{S'}(1)=2-p_1$ and $V_{S'}(2)=2-p_2$.
In Table~\ref{tab:greedy_p_prob} below, we compute the probability of each observation $B \in \mathcal{B}_1$ occurring, given the state $S'$.

\begin{table}[H]
\centering
\begin{tabular}{l|l|l|l|l}
 & $\tilde{B}_1$ & $\tilde{B}_2$ & $\tilde{B}_3$ & $\tilde{B}_4$ \\ \hline
  $\Pr{B|S'}$  & $p_1p_2$ &  $p_1(1-p_2)$ & $(1-p_1)p_2$ &   $(1-p_1)(1-p_2)$\\ 
\end{tabular}
\caption{Probability of each observation $B \in \mathcal{B}_1$ with $S'$}
\label{tab:greedy_p_prob}
\end{table}

\noindent Without loss of generality assume $p_1 < p_2$, the first product has higher value and it is always best to pick product one with value $2-p_1$. Thus, the regret of choosing the second product is, the value gap between the two products:
$\Delta=2-p_1 - (2-p_2) = p_2 - p_1 \in [0,1]$.
\newline
The payoff of the greedy strategy in this setting is:
$\pi(\sigma_g, S') = 2-\frac{p_1+p_2}{2}+\frac{(p_1-p_2)^2}{2}.$

\noindent And the regret $\bar{\gamma}$:
\[\bar{\gamma}(\sigma_g, S') = (2-p_1)-\pi(\sigma_g, S') = -\frac{p_1}{2} + \frac{p_2}{2} - \frac{(p_1-p_2)^2}{2}=\frac{\Delta}{2}-\frac{\Delta^2}{2}.\]
This is a concave quadratic in $\Delta$ on $[0,1]$, maximized at $\Delta=\frac{1}{2}$, yielding
\[\gamma(\sigma_g)=\max_{S'\in\mathcal S}\bar\gamma(\sigma_g,S')=\frac{1}{8}.\]
For two products, two ratings, after one observation and any number of states, the regret that the greedy strategy can get is at most $\frac{1}{8}$.

\noindent In comparison, the worst-case regret for the uniform strategy for $n_d=n_r=2$ and any $m \in \mathbb{N}$ is $\frac{1}{2}$, and that occurs any time $p_1=1$ and $p_2=0$ or $p_1=0$ and $p_2=1$.
\end{proof}

\noindent Next, we show that the greedy strategy is optimal and $\frac{1}{8}$ is the lowest worst-case regret any strategy can achieve for any state of Nature.

\begin{proposition}\label{proposition:greedy_optimal}
For $n_d=n_r=2$ and $m=1$, no strategy of the DM can guarantee worst-case regret strictly below $\frac18$ uniformly over all states of Nature.
\end{proposition}

\begin{proof} We prove this by contradiction, for the sake of the contradiction consider there is a fixed strategy for the DM $\sigma_1 \in \Sigma_1$ with $\gamma(\sigma_1)<\frac{1}{8}$. Consider the two states $S_1, S_2 \in \mathcal{S}$:
\[S_1 = \begin{bmatrix} 0.5 & 0 \\
0.5 & 1 \\ \end{bmatrix} \text{ and } 
S_2 = \begin{bmatrix} 0 & 0.5 \\
1 & 0.5 \\ \end{bmatrix}.\]

\noindent Consider the regret incurred for the observation matrix $\tilde{B}_4$, recall that this is the observation matrix where both products have exactly 1 two-star rating. The probability of $\tilde{B}_4$ occurring is, $\Pr{\tilde{B}_4|S_1}=\Pr{\tilde{B}_4|S_2}=\frac{1}{2}$. The value difference between the two products is $V_{S_1}(2)-V_{S_1}(1)=V_{S_2}(1)-V_{S_2}(2)=\frac{1}{2}$. For simplicity, let $p=\sigma_1(\tilde{B}_4)(1)$ and $(1-p) =\sigma_1(\tilde{B}_4)(2)$, i.e., the probabilities with which a strategy $\sigma_1$ selects product one and two respectively. Note that, $\sigma_1$ has an arbitrary but fixed value of $p$. We show that, regardless of the chosen value of $p$, either $S_1$ or $S_2$ will incur a worst-case regret larger than $\frac18$.

\noindent With $\tilde{B}_4$, for $S_1$ and $S_2$ the regret is:
\[\bar{\gamma}(\sigma_1, S_1) \geq p\Pr{\tilde{B}_4|S_1}\left(V_{S_1}(2)-V_{S_1}(1)\right)=\frac{p}{4} \text{ and}\]
\[\bar{\gamma}(\sigma_1, S_2) \geq (1-p)\Pr{\tilde{B}_4|S_2}\left(V_{S_2}(1)-V_{S_2}(2)\right)=\frac{(1-p)}{4} \text{ respectively.}\]
These equalities show the contribution from $\tilde{B_4}$; other observation matrices yield additional non-negative terms, so the total regret is at least the regret from $\tilde{B_4}$.
\noindent For the fixed strategy $\sigma_1$, by assumption the regret is smaller than $\frac18$ and thus:
\[\text{for } S_1 \text{, } \bar{\gamma}(\sigma_1, S_1)<\frac{1}{8} \text{ or } \frac{p}{4}<\frac{1}{8} \rightarrow p < \frac{1}{2},\]
\[\text{and for } S_2 \text{, }\bar{\gamma}(\sigma_1, S_2)<\frac{1}{8} \text{ or }\frac{(1-p)}{4}<\frac{1}{8}\rightarrow p > \frac{1}{2}.\]
This contradiction yields the claim. 
\end{proof}

\begin{corollary}[Minimax optimality of greedy for $n_d=n_r=2$ and $m=1$]\label{cor:greed_optimal}
In the setting $n_d=n_r=2$ and $m=1$, the greedy strategy $\sigma_g$ is minimax-optimal:
\[
\gamma(\sigma_g)=\inf_{\sigma_1\in\Sigma_1}\gamma(\sigma_1)=\frac18.
\]
\end{corollary}
\begin{proof}
Proposition \ref{proposition:regretm1} shows $\gamma(\sigma_g)=\frac{1}{8}$, and Proposition \ref{proposition:greedy_optimal} shows that every strategy satisfies $\gamma(\sigma_1) \ge \frac{1}{8}$. Combining the two yields the claim.
\end{proof}

\subsection{Numerical evaluation of greedy's regret for \texorpdfstring{$2\leq m \leq 20$}{2≤m≤20} observations per product }\label{section:regret_m12}

Here we increase the number of observations up to $m=20$, but still consider two products and two ratings. For each value of $m$, both products
receive exactly $m$ observations. In order to calculate what is the maximum regret with the greedy strategy as the observations increase, we express $\bar\gamma$ with respect to $p_1$ and $p_2$ from the $S'$ state in Section \ref{section:regret_m1}, and numerically solve for its maximum.

For $n_d=n_r=2$ and for any state $S \in \mathcal{S}$, Figure~\ref{fig:numresult} presents a numerical bound for the worst-case regret of the greedy strategy as $m$ increases up to 20. Notice that the regret is quickly decreasing, after 10 observations the maximum regret would have decreased at least threefold from $\frac{1}{8}$ down to $\approx\frac{1}{24}$.
This monotonic decrease should be read as an equal-observation result. In Section~\ref{sec:het_obs_res}, we show that for the different number of observations case greedy can behave differently when observations are added for only one product.

\begin{figure}[H]
\caption{Maximum regret of the greedy strategy choosing between two products as the number of observations increase up to 20}\label{fig:numresult}
\centering
\includegraphics[scale=0.51]{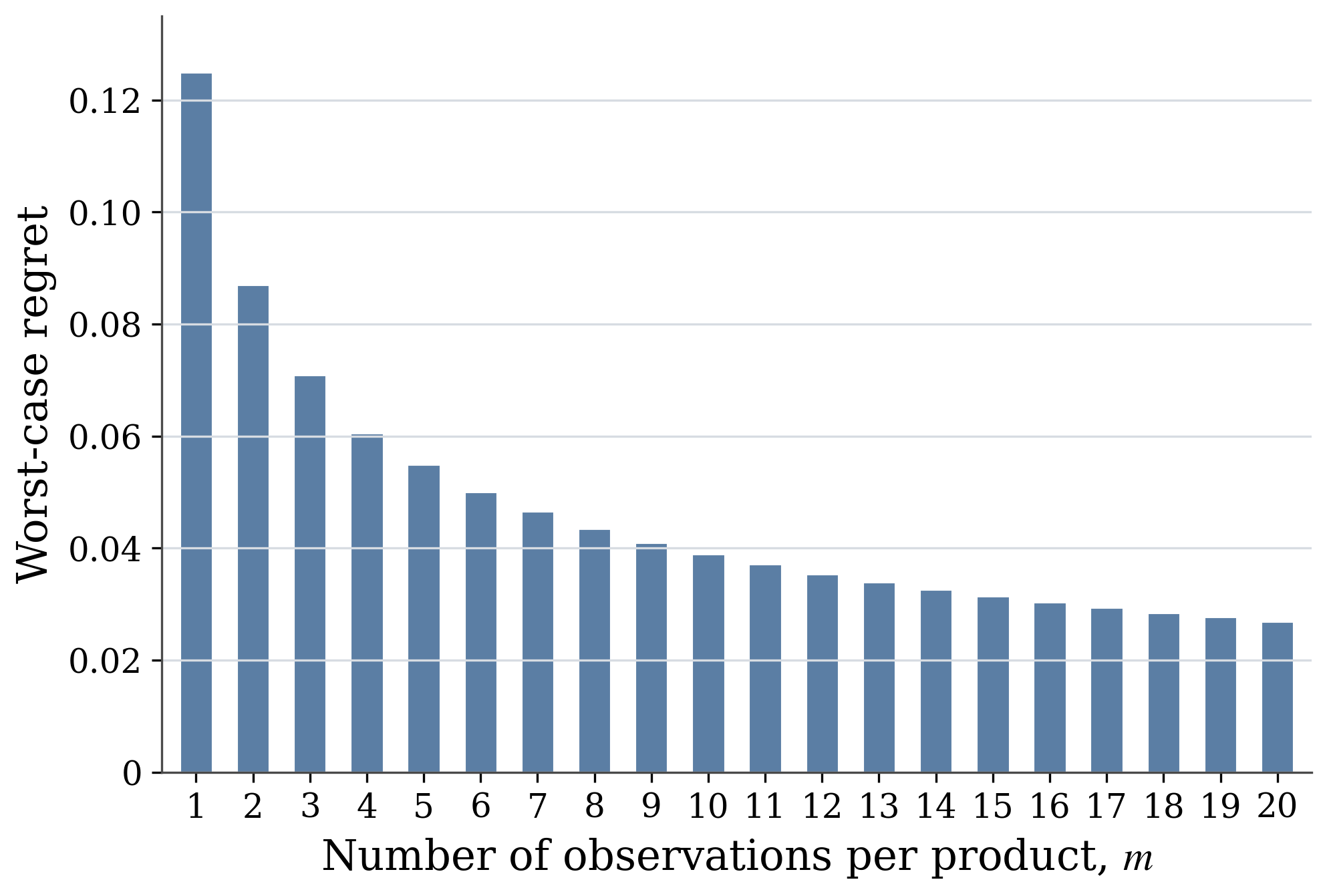}
\end{figure}

\section{Results: Heterogeneous Observations}\label{sec:het_obs_res}
\noindent In the previous sections, we assumed that every product has the same number of observations. This is a balanced-information benchmark: the greedy strategy compares empirical values that are estimated from the same amount of data. We now relax this assumption and allow the number of observations to differ across products. Let
\[
\mathbf m=(m_d)_{d\in D}
\]
denote the observation vector, where $m_d\geq 1$ is the number of observations available for product $d$, and define
\[
m_{\min}=\min_{d\in D}m_d .
\]
Throughout this section, payoff and regret are computed using this observation vector. The greedy strategy is unchanged: it still selects a product with the highest observed average rating. \newline
For a fixed observation vector $\mathbf m$, the set of possible observation
matrices is
\[ \mathcal B_{\mathbf m} = \left\{ B\in \mathbb N_0^{n_r\times n_d} :
\sum_{r\in R} b_{r,d}=m_d \text{ for all } d\in D \right\}.
\]
For $B\in\mathcal B_{\mathbf m}$, we keep the notation $V_B(d)$ for the
observed value of product $d$, now computed using the number of observations
available for that product:
\[
V_B(d)=\frac{1}{m_d}\sum_{r\in R} r b_{r,d}.
\]

\subsection{Greedy's upper bound regret with heterogeneous observations}
In this section we show that the finite-sample upper bound for greedy is robust to unequal numbers of observations. The bound is conservative: it depends only on the least observed product.

\begin{proposition}\label{prop:het_up_bound}
The greedy strategy satisfies
\[
\gamma(\sigma_g)
\leq
(n_r-1)\sqrt{\frac{2\log(n_d)}{m_{\min}}}.
\]
\end{proposition}

\begin{proof}
Fix a state $S\in\mathcal S$. For each product $d$, the proof of
Lemma~\ref{lem:concentration-VB} applies with $m_d$ observations instead of $m$. Hence, for every $t\geq 0$,
\[
\mathbb P\left[ |V_B(d)-V_S(d)|\geq t \mid S \right] \leq
2\exp\left( -\frac{2m_dt^2}{(n_r-1)^2} \right).
\]
Since $m_d\geq m_{\min}$, we also have
\[
\mathbb P\left[ |V_B(d)-V_S(d)|\geq t \mid S \right] \leq 
2\exp\left( - \frac{2m_{\min}t^2}{(n_r-1)^2} \right).
\]
Thus the empirical errors satisfy the same sub-Gaussian bound used in the proof of Proposition~\ref{thm:greedy-upper}, with $m_{\min}$ in place of
$m$.
The regret decomposition in the proof of Proposition~\ref{thm:greedy-upper} is unchanged, since it only uses the fact that greedy selects a product with the highest observed value. Applying the same expected-maximum bound with $m_{\min}$ gives, for every state $S$,
\[
\bar\gamma(\sigma_g,S) \leq (n_r-1)\sqrt{\frac{2\log(n_d)}{m_{\min}}}.
\]
Taking the maximum over $S\in\mathcal S$ completes the proof.
\end{proof}

\paragraph{Remark on the lower-bound argument with different number of observations.}
The lower-bound proof from Proposition~\ref{lem:minimax-lower} does not extend by simply replacing $m$ with $m_{\min}$. The upper bound above only uses concentration of each product's empirical value, and therefore the least number of observations is sufficient. The lower-bound proof constructs states that are hard to distinguish from the observations. In the equal number of observations case, the KL divergence between two hard states from Lemma~\ref{lem:kl_pairwise} is proportional to $m$, because every product has $m$ observations. With heterogeneous observations, if two hard states differ in products $i$ and $j$, the corresponding KL term is instead proportional to $m_i+m_j$. Thus the lower-bound scale depends on the particular observation vector, or on the group of products over which the hard-state construction is applied. We therefore use the $m_{\min}$ result only as a conservative upper bound, and next study the exact heterogeneous behavior of greedy in the two-product case.

\subsection{One-sided Observations Can Increase Greedy's Regret}\label{sec:greedy_het_one}
\noindent The previous result gives a regret upper bound for greedy under heterogeneous observations. We now show that the exact worst-case regret can behave in a qualitatively different way. At first sight, the phenomenon is paradoxical: adding observations for one product can increase the worst-case regret of greedy. The reason is that greedy compares observed averages without accounting for the different number of observations (precision) of those averages. A product observed once can attain the maximum empirical value after a single high rating. By contrast, a product observed many times is likely to reveal even rare low ratings, so its empirical average may fall below the maximum possible value. Thus, in worst-case states, additional observations can break a tie that previously protected greedy and can make the less observed product strictly preferred. \newline
Throughout this subsection, we use the two-product, two-rating state $S'$ from Section~\ref{section:regret_m1}, where $p_d$ denotes the probability that product $d$ receives rating $1$. Hence $V_{S'}(d)=2-p_d$.
We compare the equal-observations vector $\mathbf m=(1,1)$ with the one-sided path $\mathbf m=(1,k)$, where product $1$ is observed once and product $2$ is observed $k$ times.

\begin{proposition}\label{prop:ones_regret}
Suppose $n_d=n_r=2$. Then
\[
\left.
\gamma(\sigma_g)
\right|_{\mathbf m=(1,1)}
=
\frac18,
\qquad
\lim_{k\to\infty}
\left.
\gamma(\sigma_g)
\right|_{\mathbf m=(1,k)}
=
\frac14 .
\]
The same limit holds along the path $\mathbf m=(k,1)$.
\end{proposition}

\begin{proof}
The equality for $\mathbf m=(1,1)$ follows from Proposition~\ref{proposition:regretm1}. We prove the one-sided limit. Under $\mathbf m=(1,k)$, product $1$ is observed once, while product $2$ is observed $k$ times. First suppose $p_1\leq p_2$. Then product $1$ is optimal, and the value gap is $V_{S'}(1)-V_{S'}(2)=p_2-p_1$.
Regret occurs when greedy selects product $2$. If product $1$'s single observation is rating $1$, which occurs with probability $p_1$, then product $1$ has the minimum empirical value. Greedy selects product $2$ unless all $k$ observations of product $2$ are also rating $1$, in which case greedy randomizes. This gives the term $p_1\left(1-\frac{p_2^k}{2}\right)$. If product $1$'s single observation is rating $2$, which occurs with probability $1-p_1$, then product $1$ has the maximum empirical value. Greedy selects product $2$ only if all $k$ observations of product $2$ are also rating $2$, in which case greedy again randomizes. This gives the term $\frac{1-p_1}{2}(1-p_2)^k$. Therefore, in the case $p_1\leq p_2$,
\[
\left.
\bar\gamma(\sigma_g,S')
\right|_{\mathbf m=(1,k)}
=
(p_2-p_1)
\left[
p_1\left(1-\frac{p_2^k}{2}\right)
+
\frac{1-p_1}{2}(1-p_2)^k
\right].
\]
Since $\frac{p_1 p_2^k}{2} \geq 0 $, we have $p_1 - \frac{p_1 p_2^k}{2} \leq p_1$, thus we can bound the regret expression by
\[
\bar\gamma(\sigma_g,S')|_{\mathbf m=(1,k)} \leq  p_1(p_2-p_1)
+
\frac12(p_2-p_1)(1-p_1)(1-p_2)^k .
\]
For the first term, $\max_{0\leq p_1\leq p_2\leq 1}p_1(p_2-p_1)=\frac14$, thus:
\[\left. \bar\gamma(\sigma_g,S') \right|_{\mathbf m=(1,k)} \leq \frac{1}{4} + \frac12(p_2-p_1)(1-p_1)(1-p_2)^k .\]
For the second term, since $p_2-p_1\leq p_2$ and $1-p_1\leq 1$,
\[\left. \bar\gamma(\sigma_g,S') \right|_{\mathbf m=(1,k)} \leq \frac{1}{4} + \frac12 p_2(1-p_2)^k .\]
The remaining term is now of the form  $f(x)=x(1-x)^k$ over $x\in[0,1]$. Then, $f'(x)=(1-x)^{k-1}(1-(k+1)x)$, so the maximum is attained at $x=\frac{1}{(k+1)}$, hence
\[
f\left(\frac{1}{k+1}\right)
=
\frac{1}{k+1}\left(\frac{k}{k+1}\right)^k
\leq
\frac{1}{k+1}.
\]
For the case of $p_1\leq p_2$,
\[\left. \bar\gamma(\sigma_g,S') \right|_{\mathbf m=(1,k)} \leq \frac{1}{4} + \frac{1}{2(k+1)} .\]
Now suppose $p_2\leq p_1$. Then product $2$ is optimal, and the value gap is $V_{S'}(2)-V_{S'}(1)=p_1-p_2$. By the same enumeration as above, regret occurs when greedy selects product $1$, and
\[\left. \bar\gamma(\sigma_g,S') \right|_{\mathbf m=(1,k)} = (p_1-p_2)
\left[
(1-p_1)\left(1-\frac{(1-p_2)^k}{2}\right)
+
\frac{p_1}{2}p_2^k
\right].
\]
Since $\frac{(1-p_1)(1-p_2)^k}{2}\geq 0$, we have
$(1-p_1)-\frac{(1-p_1)(1-p_2)^k}{2}\leq 1-p_1$, thus we can bound the regret expression by
\[
\left.
\bar\gamma(\sigma_g,S')
\right|_{\mathbf m=(1,k)}
\leq
(1-p_1)(p_1-p_2)
+
\frac12 p_1(p_1-p_2)p_2^k .
\]
For the first term, $\max_{0\leq p_2\leq p_1\leq 1}(1-p_1)(p_1-p_2)=\frac14$, thus:
\[
\left.
\bar\gamma(\sigma_g,S')
\right|_{\mathbf m=(1,k)}
\leq
\frac14
+
\frac12 p_1(p_1-p_2)p_2^k .
\]
For the second term, since $p_1\leq 1$ and $p_1-p_2\leq 1-p_2$,
\[
\left.
\bar\gamma(\sigma_g,S')
\right|_{\mathbf m=(1,k)}
\leq
\frac14
+
\frac12(1-p_2)p_2^k .
\]
The remaining term is now of the form $f(x)=(1-x)x^k$ over $x\in[0,1]$. Then $f'(x)=x^{k-1}(k-(k+1)x)$, so the maximum is attained at $x=\frac{k}{k+1}$, hence
\[
f\left(\frac{k}{k+1}\right)
=
\frac{1}{k+1}
\left(\frac{k}{k+1}\right)^k
\leq
\frac{1}{k+1}.
\]
For the case of $p_2\leq p_1$,
\[
\left.
\bar\gamma(\sigma_g,S')
\right|_{\mathbf m=(1,k)}
\leq
\frac14+\frac{1}{2(k+1)}.
\]
Since both of the cases have the same upper bound, taking the maximum over states gives
\[
\limsup_{k\to\infty}
\left.
\gamma(\sigma_g)
\right|_{\mathbf m=(1,k)}
\leq
\frac14 .
\]
It remains to show that this upper bound can be approached. For large $k$, set
\[
p_1=\frac12,
\qquad
p_2=\varepsilon_k,
\qquad
\varepsilon_k=\frac{1}{\sqrt{k}} .
\]
For a given number of observations for the second product - $k$, we define a state $S'_k$. Product $2$ is optimal, since it receives rating $1$ with probability $\varepsilon_k$, while product $1$ receives rating $1$ with probability $\frac{1}{2}$. The value gap is
\[ V_{S'_k}(2)-V_{S'_k}(1) = p_1-p_2 = \frac12-\varepsilon_k.\]
The choice $\varepsilon_k=\frac{1}{\sqrt{k}}$ serves two purposes. First, $\varepsilon_k\to 0$, so product $2$ becomes almost always high-rated. Second, $k\varepsilon_k\to\infty$, so even though low ratings for product $2$ are rare, at least one low rating appears among $k$ observations of product $2$ with probability tending to one. Indeed,
\[(1-\varepsilon_k)^k \leq e^{-k\varepsilon_k} = e^{-\sqrt{k}} \to 0.\]
Now consider the event that product $1$'s single observation is rating $2$, while product $2$'s $k$ observations are not all rating $2$. This event has probability
\[(1-p_1)\left[1-(1-p_2)^k\right] = \frac12\left[1-(1-\varepsilon_k)^k\right].\]
On this event, product $1$ has empirical value $2$, while product $2$ has empirical value strictly below $2$. Therefore greedy selects product $1$, the inferior product. Hence
\[ \left. \bar\gamma(\sigma_g,S'_k) \right|_{\mathbf m=(1,k)} \geq \left(\frac12-\varepsilon_k\right) \frac12 \left[1-(1-\varepsilon_k)^k\right]. \] 
Since $\gamma(\sigma_g)$ is the maximum regret over states,
\[
\liminf_{k\to\infty}
\left.
\gamma(\sigma_g)
\right|_{\mathbf m=(1,k)}
\geq
\frac14 .
\]
Together with the upper bound, this proves
\[
\lim_{k\to\infty}
\left.
\gamma(\sigma_g)
\right|_{\mathbf m=(1,k)}
=
\frac14 .
\]
The result for $\mathbf m=(k,1)$ follows by relabeling the products.
\end{proof}

\noindent Using the same numerical procedure as in Section~\ref{section:regret_m12}, we compute the worst-case regret of greedy along the one-sided heterogeneous observation path
$\mathbf m=(m_1,1)$ with $m_1 \in \{1,2,\cdots,30\}$. That is, for each fixed pair $(m_1,1)$, we solve the corresponding maximization problem over states and record the largest regret attained by the greedy strategy. This isolates the heterogeneous path in which the number of observations increases for only one product, while the other product remains observed once. Figure~\ref{fig:heterogeneous-greedy-regret} shows that this one-sided path behaves differently from the equal-observations path in Section~\ref{section:regret_m12}. In the equal-observations case, regret decreases as both products receive more observations. Along the one-sided path, however, regret can increase: adding observations for only one product may make greedy perform worse in the worst case.

\begin{figure}[H]
\caption{Worst-case regret of the greedy strategy in the two-product, two-rating case as a function of the heterogeneous observation vector
$\mathbf m=(m_1,1)$.}\label{fig:heterogeneous-greedy-regret}
\centering
\includegraphics[scale=0.37]{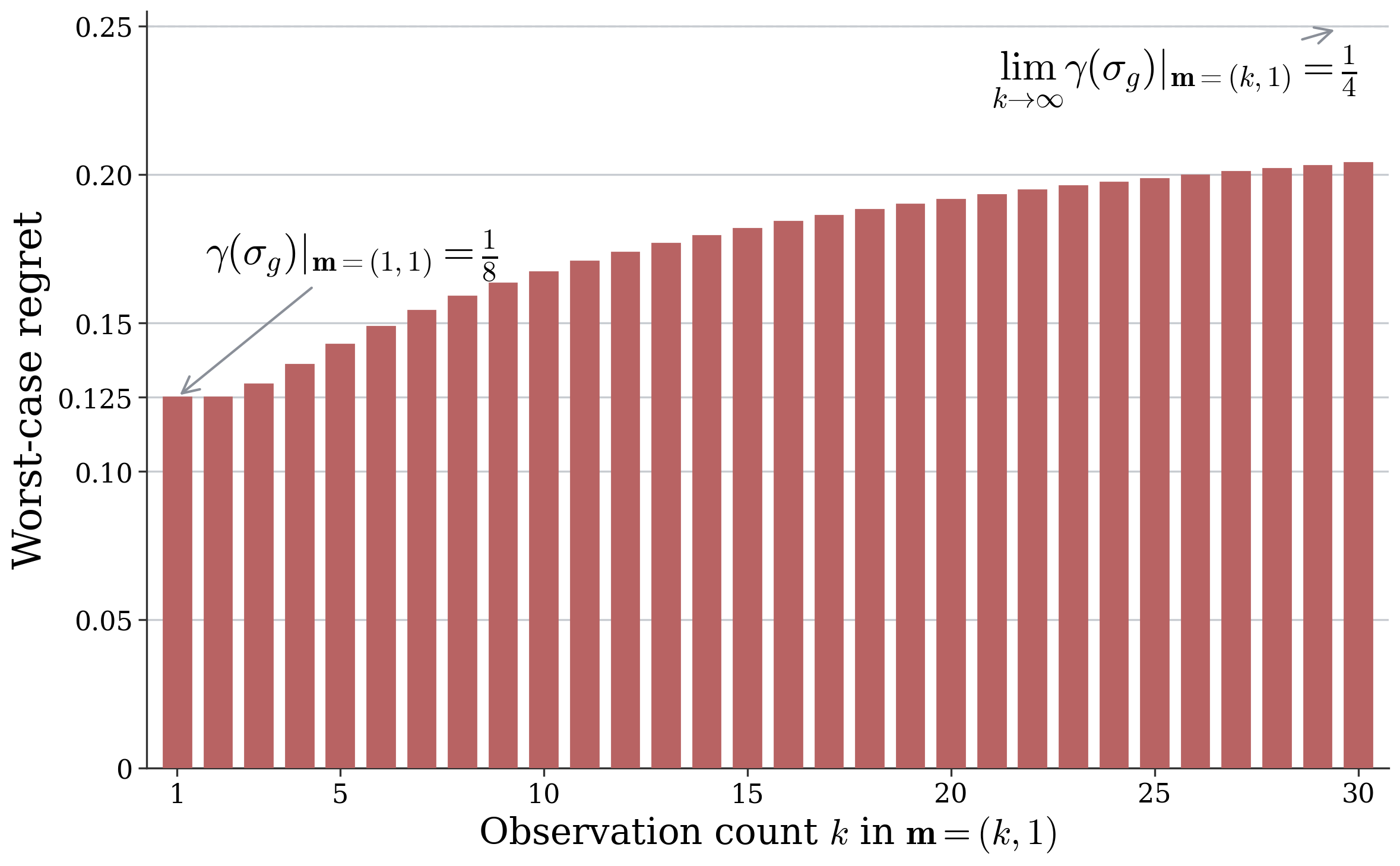}
\end{figure}
\noindent The proposition shows that the nonmonotonicity is not caused by information having negative value. It is a property of the greedy strategy. Greedy uses the observed averages directly, even when they are based on different numbers of observations. In the worst-case states above, the additional observations make the better product's rare low ratings visible, while the less observed product can still appear perfect after a single high rating.
This result does not contradict Proposition~\ref{prop:het_up_bound}. Along the one-sided observations path $\mathbf m=(k,1)$, the least observation count remains $m_{\min}=1$. Hence the conservative upper bound from Proposition~\ref{prop:het_up_bound} does not improve along this path. The nonmonotonicity concerns the exact worst-case regret of the greedy strategy when observations are added to only one product.

\subsection{Confidence Bound Algorithms In The One-Shot Model}\label{sect:UCB_het}
\noindent 
The previous subsection shows that greedy can behave nonmonotonically under heterogeneous observations. A natural response is to use a confidence-bound algorithm, since such algorithms explicitly account for unequal sample sizes. We show that there exist states of Nature, such that with heterogeneous observations, UCB has regret that is twice as large as the greedy strategy. For an empirical value of a product $d$ - $V_B(d)$, any UCB score of the form
\[
I_d(B)=V_B(d)+\beta(m_d)
\]
coincides with greedy whenever $m_d=m$ for all $d\in D$, because the confidence term $\beta(m_d)$ is constant across products. With unequal observations, however, a strictly decreasing confidence term $\beta(\cdot)$ gives less observed products a larger score. This is natural in sequential bandit problems, where selecting (exploring) a less observed product produces new information used in future rounds. In our model, the decision is one-shot, so the information from exploring a less-observed product does not get used. This distinction can increase worst-case regret. The following example isolates this effect and shows that UCB can incur strictly greater regret than the greedy strategy. Consider $n_d=n_r=2$ and $\mathbf m=(1,2)$. Thus $\beta(1)>\beta(2)$, and consider the state
\[ S^{\mathrm{UCB}} =
\begin{bmatrix}
\frac12 & 0\\
\frac12 & 1
\end{bmatrix}.
\]
Thus,
\[ V_{S^{\mathrm{UCB}}}(2)-V_{S^{\mathrm{UCB}}}(1) = 2-\frac32 = \frac12 .\]
Product $2$ is strictly better and it always receives rating $2$; under $\mathbf m=(1,2)$ its observed value is always $V_B(2)=2$. Consider the following observation matrix: \[B^{UCB}=\begin{bmatrix}
0 & 0\\
1 & 2
\end{bmatrix}.\]
Observe that, although product $2$ is strictly better, both of the products have the same observed value for this observation matrix - $V_{B^{UCB}}(1)=V_{B^{UCB}}(2)=2$. Product $1$ is observed once, and with probability $\frac{1}{2}$ this observation is also rating $2$, thus $\Pr{B^{UCB}|S^{UCB}}=\frac12$. Although the observed values are tied, the UCB scores are not:
\[ I_1(B) = 2+\beta(1) > 2+\beta(2) = I_2(B). \]
Hence UCB selects product $1$, which is the inferior product. Therefore
\[
\left. \bar\gamma(\sigma_{\mathrm{UCB}},S^{\mathrm{UCB}}) \right|_{\mathbf m=(1,2)} \geq \left( V_{S^{UCB}}(2)-V_{S^{UCB}}(1) \right)\Pr{B^{UCB}|S^{UCB}}\sigma_{UCB}(B^{UCB})(1) \geq \frac14 .
\]
The only other observation matrix with positive probability under $S^{UCB}$ is $\begin{bmatrix}
1 & 0\\
0 & 2
\end{bmatrix}$, here product two has the larger empirical value, so greedy selects product two and incurs no regret. UCB also selects product 2 when $\beta(1)-\beta(2)<1$; for a larger confidence-bonus difference, it may incur additional regret at this observation matrix. In either case, this event does not reduce the preceding lower bound.
Taking the maximum over states gives
\[ \left. \gamma(\sigma_{\mathrm{UCB}}) \right|_{\mathbf m=(1,2)} \geq \frac14 .
\]
At the same state, greedy can incur regret only at $B^{UCB}$. Since this observation matrix occurs with probability $\frac12$, and greedy selects the inferior product with probability $\frac12$ under uniform tie-breaking,
\[ \left. \bar\gamma(\sigma_g,S^{\mathrm{UCB}}) \right|_{\mathbf m=(1,2)} = \left( V_{S^{UCB}}(2)-V_{S^{UCB}}(1) \right)\Pr{B^{UCB}|S^{UCB}}\sigma_g(B^{UCB})(1)
=
\frac18.
\]
Moreover, the greedy numerical regret maximization in Section \ref{sec:greedy_het_one} for $\mathbf m=(1,2)$ over all possible states $S \in \mathcal{S}$ gives:
\[
\left.
\gamma(\sigma_g)
\right|_{\mathbf m=(1,2)}
=
\max_{S\in\mathcal S}
\left.
\bar{\gamma}(\sigma_g,S)
\right|_{\mathbf m=(1,2)}
=
\frac18.
\]
We conclude that for the specific observations vector $\mathbf m=(1,2)$, UCB incurs at least twice the regret of greedy: $\frac{1}{4}$ instead of $\frac{1}{8}$. The confidence term is harmful here because there is no future round in which the information from exploring the less-observed product can be used. \newline
This result is not specific to upper confidence bounds. An analogous counterexample applies to lower confidence bound algorithms (LCB), which score is $V_B(d)-\beta(m_d)$. By symmetrically reversing the rating rows of $S^{\mathrm{UCB}}$ and $B^{\mathrm{UCB}}$, the observed values remain tied, but that makes the more-observed product inferior, which LCB selects and thereby incurs regret at least $\frac14$.

\section{Empirical Results for the Greedy Strategy on Google Reviews Data for Restaurants}
\label{sec:Google}

We next give a finite-sample empirical illustration using Google reviews data for restaurants \cite{YHLZM2022}. The objective is to select a restaurant with high consumer satisfaction, independent of its cuisine. This experiment is an empirical counterpart to our theoretical results for the equal number of observations case, showing that the greedy strategy performs well even when the state of Nature is unknown. The dataset contains 1.5 million reviews for 64,000 restaurants, each with a single review per guest and ratings in $R = \{1,2,3,4,5\}$. We only consider restaurants with at least 10 observations. The average rating is 4.45 with a standard deviation of 1. Given the dataset's size, we assume Nature has fixed its state $S \in \mathcal{S}$, and that the reviews reflect this state. We treat each restaurant's average rating across all its reviews as its true value and benchmark the greedy and Thompson Sampling strategies against this ground truth. Thompson Sampling uses the symmetric Dirichlet prior $\alpha_0=\mathbf 1$. The experiment proceeds as follows:

\begin{enumerate}
\item Sample $n_d$ restaurants uniformly from the set of restaurants and calculate their value, by taking the average across all available observations for each restaurant.
\item Draw uniformly at random without replacement $m$ different observations (reviews) for each restaurant.
\item Run each strategy on the sampled observations and record which restaurants they pick.
\item Calculate the regret between the best restaurant from step 1 and the payoff from the restaurants that step 3 selects.
\end{enumerate}

\noindent For $n_d\in\{2,3,\ldots,10\}$, $m\in\{1,2,\ldots,10\}$, and $n_r=5$, we evaluate all three strategies over 1500 simulations for each pair $(n_d,m)$. Greedy attains lower mean regret than both the uniform strategy and the Thompson Sampling algorithm for every tested pair $(n_d,m)$. Its mean regret declines overall as the number of observations increases for each number of products, whereas the regret of uniform selection as expected is insensitive to the number of observations. These results are depicted in Figure~\ref{fig:simuresult}.

\begin{figure}[H]
\caption{Regret of the proposed strategies when tested on Google reviews data for restaurants. The regret of the greedy strategy is always below all other strategies.}\label{fig:simuresult}
\centering
\includegraphics[scale=0.41]{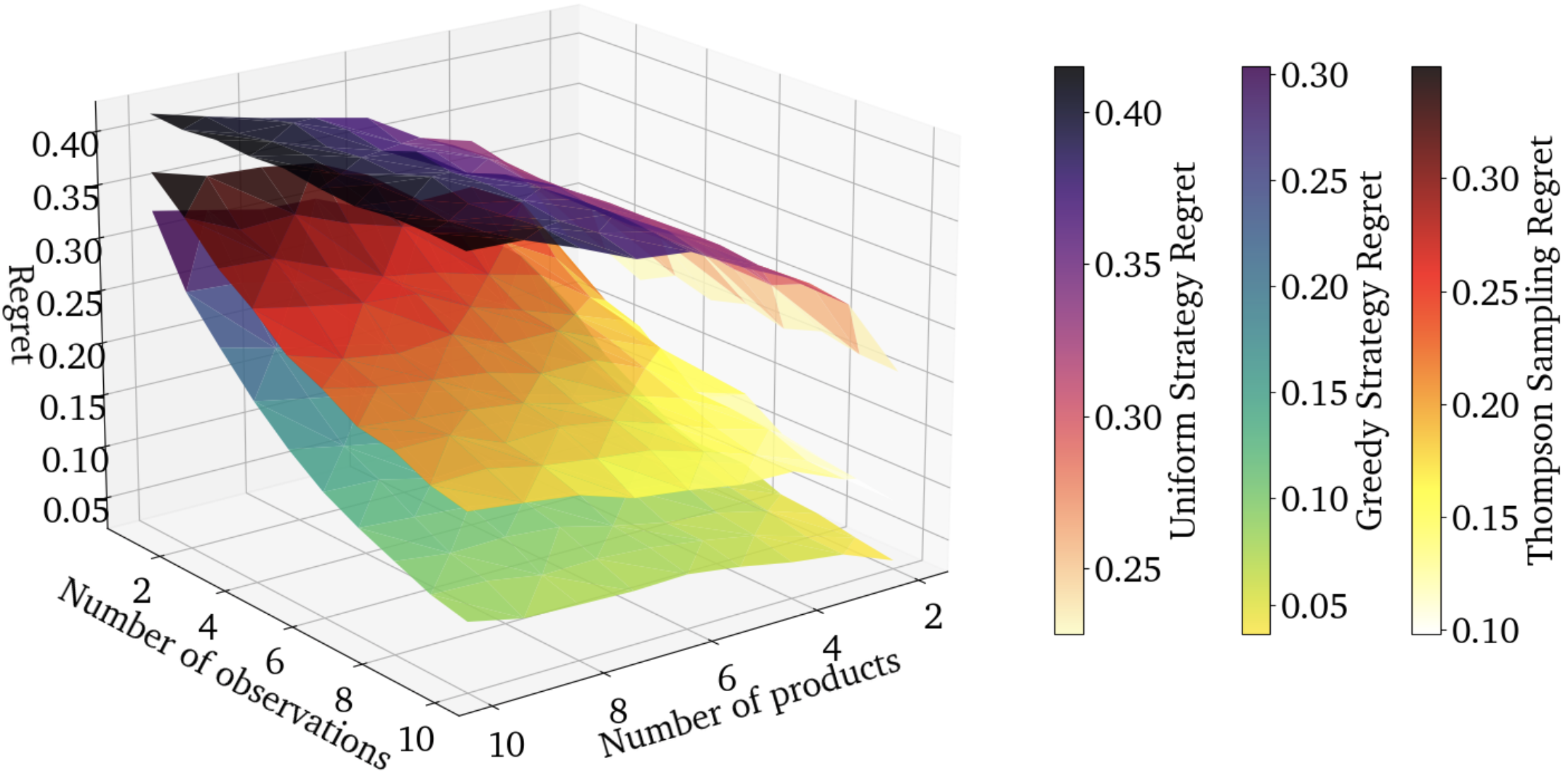}
\end{figure}
\noindent For enhanced clarity, tabulated results for Figure \ref{fig:simuresult} can be found in Appendix \ref{section:appendix_tabulated_google}.

\section{Conclusions}
\noindent We studied decision-making under complete (Knightian) uncertainty via a game between a decision maker and Nature, and analyzed the greedy strategy that chooses the product with the highest observed average rating. Our main theoretical contribution is an equal-observation benchmark: when every product has the same number of observations, we provide a finite-sample upper bound on greedy's worst-case regret and a matching lower bound for all strategies. Thus, in the equal-observations model, greedy achieves the minimax regret rate up to universal constants. In the two-product, two-rating case with one observation per product, we further characterize the exact minimax value: greedy's worst-case regret is $\frac{1}{8}$, and no strategy can do better. \newline
We also studied heterogeneous observations, where products may have different numbers of reviews. The greedy upper bound extends conservatively with the common sample size replaced by $m_{\min}$, the least number of observations across products. However, the exact behavior is no longer monotone in the same way as in the equal number of observations case. In the two-product, two-rating case, adding observations for only one product can increase greedy's worst-case regret, and the one-sided path $\mathbf m=(1,k)$ converges to regret $\frac{1}{4}$. \newline
These results give two complementary messages. When sample sizes are equal across options, choosing by the observed average rating is not only computationally simple but also minimax-rate optimal in our model. When sample sizes are heterogeneous, greedy remains controlled by a conservative $m_{\min}$ upper bound, but one-sided increases in observations can change the exact worst-case behavior.

\newpage
\bibliographystyle{plain} 
\bibliography{refs}

\newpage
\appendix

\section{Decision maker's strategy lower bound regret proofs}
\subsection{KL divergence equivalence for the Bernoulli and Binomial distribution for the two product case}\label{section:kl_equals}
The following is folklore and we include it for the sake of completeness. 
\begin{lemma}[Binomial KL equals $m$ times Bernoulli KL]\label{lem:binom_bern_kl}
For any $m\in\mathbb{N}$ and any $p,q\in(0,1)$,
\[
\mathrm{KL}(\mathrm{Binomial}(m,p)\|\mathrm{Binomial}(m,q))
=m\,\mathrm{KL}(\mathrm{Bern}(p)\|\mathrm{Bern}(q)).
\]
\end{lemma}

\begin{proof}
For $x\in\{0,1,\ldots,m\}$,
\[
\mathrm{Binomial}(m,p)(x)=\binom{m}{x}p^x(1-p)^{m-x}.
\]
Hence
\[
\mathrm{KL}(\mathrm{Binomial}(m,p)\|\mathrm{Binomial}(m,q))
=\sum_{x=0}^m\binom{m}{x}p^x(1-p)^{m-x}\log \left(\frac{p^x(1-p)^{m-x}}{q^x(1-q)^{m-x}}\right).
\]
The binomial coefficient cancels inside the logarithm, giving
\[
=\sum_{x=0}^m\binom{m}{x}p^x(1-p)^{m-x}\left(x\log\left(\frac{p}{q}\right)+(m-x)\log\left(\frac{1-p}{1-q}\right)\right).
\]
Let $X\sim\mathrm{Binomial}(m,p)$.
Then $\mathbb{E}[X]=mp$ and $\mathbb{E}[m-X]=m(1-p)$, so the last display equals
\[
mp\log\left(\frac{p}{q}\right)+m(1-p)\log\left(\frac{1-p}{1-q}\right)
=m\left(p\log\left(\frac{p}{q}\right)+(1-p)\log \left(\frac{1-p}{1-q}\right)\right)
=m\,\mathrm{KL}(\mathrm{Bern}(p)\|\mathrm{Bern}(q)).
\]
\end{proof}

\section{Tabulated Results for Google Reviews Experiments}\label{section:appendix_tabulated_google}
Tabulated results are provided in this section to complement the result from Figure \ref{fig:simuresult}. Note, row indices are number of observations, and column indices are number of products and per the Google reviews standard the number of rating levels for all experiments is $n_r=5$.

\begin{table}[h]
\centering
\resizebox{\textwidth}{!}{
\begin{tabular}{|c|c|c|c|c|c|c|c|c|c|}
\hline
$m,n_d$ & 2 & 3 & 4 & 5 & 6 & 7 & 8 & 9 & 10 \\ \hline
1  & $0.14 \pm 0.22$ & $0.20 \pm 0.20$ & $0.24 \pm 0.20$ & $0.26 \pm 0.17$ & $0.28 \pm 0.16$ & $0.31 \pm 0.16$ & $0.31 \pm 0.15$ & $0.32 \pm 0.13$ & $0.33 \pm 0.13$ \\ \hline
2  & $0.11 \pm 0.20$ & $0.14 \pm 0.18$ & $0.18 \pm 0.18$ & $0.21 \pm 0.17$ & $0.22 \pm 0.17$ & $0.24 \pm 0.15$ & $0.25 \pm 0.14$ & $0.26 \pm 0.14$ & $0.27 \pm 0.13$ \\ \hline
3  & $0.09 \pm 0.17$ & $0.12 \pm 0.17$ & $0.15 \pm 0.16$ & $0.17 \pm 0.16$ & $0.18 \pm 0.15$ & $0.19 \pm 0.15$ & $0.20 \pm 0.15$ & $0.21 \pm 0.15$ & $0.23 \pm 0.14$ \\ \hline
4  & $0.06 \pm 0.14$ & $0.11 \pm 0.16$ & $0.12 \pm 0.15$ & $0.13 \pm 0.16$ & $0.15 \pm 0.15$ & $0.16 \pm 0.15$ & $0.17 \pm 0.14$ & $0.18 \pm 0.14$ & $0.19 \pm 0.15$ \\ \hline
5  & $0.05 \pm 0.12$ & $0.09 \pm 0.14$ & $0.11 \pm 0.15$ & $0.12 \pm 0.14$ & $0.14 \pm 0.15$ & $0.14 \pm 0.14$ & $0.14 \pm 0.14$ & $0.15 \pm 0.13$ & $0.16 \pm 0.15$ \\ \hline
6  & $0.06 \pm 0.13$ & $0.07 \pm 0.13$ & $0.09 \pm 0.14$ & $0.10 \pm 0.13$ & $0.11 \pm 0.14$ & $0.13 \pm 0.14$ & $0.12 \pm 0.13$ & $0.14 \pm 0.14$ & $0.14 \pm 0.13$ \\ \hline
7  & $0.04 \pm 0.11$ & $0.06 \pm 0.12$ & $0.08 \pm 0.12$ & $0.09 \pm 0.11$ & $0.10 \pm 0.13$ & $0.10 \pm 0.12$ & $0.11 \pm 0.12$ & $0.12 \pm 0.13$ & $0.12 \pm 0.12$ \\ \hline
8  & $0.04 \pm 0.11$ & $0.06 \pm 0.11$ & $0.07 \pm 0.11$ & $0.07 \pm 0.11$ & $0.08 \pm 0.11$ & $0.10 \pm 0.12$ & $0.10 \pm 0.12$ & $0.10 \pm 0.11$ & $0.11 \pm 0.12$ \\ \hline
9  & $0.03 \pm 0.09$ & $0.05 \pm 0.10$ & $0.06 \pm 0.10$ & $0.07 \pm 0.11$ & $0.08 \pm 0.11$ & $0.08 \pm 0.11$ & $0.08 \pm 0.11$ & $0.09 \pm 0.11$ & $0.09 \pm 0.11$ \\ \hline
10 & $0.03 \pm 0.07$ & $0.05 \pm 0.09$ & $0.05 \pm 0.10$ & $0.06 \pm 0.10$ & $0.07 \pm 0.10$ & $0.07 \pm 0.09$ & $0.07 \pm 0.10$ & $0.07 \pm 0.10$ & $0.08 \pm 0.10$ \\ \hline
\end{tabular}
}
\caption{Mean regret $\pm$ standard deviation of the greedy strategy on the Google Reviews data for restaurants.}
\label{tab:google-greedy-std}
\end{table}

\begin{table}[h]
\centering
\resizebox{\textwidth}{!}{
\begin{tabular}{|c|c|c|c|c|c|c|c|c|c|}
\hline
$m,n_d$ & 2 & 3 & 4 & 5 & 6 & 7 & 8 & 9 & 10 \\ \hline
1  & $0.18 \pm 0.33$ & $0.26 \pm 0.36$ & $0.28 \pm 0.37$ & $0.28 \pm 0.34$ & $0.32 \pm 0.34$ & $0.34 \pm 0.37$ & $0.35 \pm 0.36$ & $0.35 \pm 0.34$ & $0.36 \pm 0.37$ \\ \hline
2  & $0.16 \pm 0.32$ & $0.22 \pm 0.33$ & $0.24 \pm 0.30$ & $0.27 \pm 0.32$ & $0.28 \pm 0.33$ & $0.30 \pm 0.33$ & $0.32 \pm 0.35$ & $0.31 \pm 0.32$ & $0.33 \pm 0.32$ \\ \hline
3  & $0.15 \pm 0.27$ & $0.19 \pm 0.28$ & $0.22 \pm 0.29$ & $0.24 \pm 0.28$ & $0.25 \pm 0.29$ & $0.26 \pm 0.28$ & $0.28 \pm 0.30$ & $0.30 \pm 0.30$ & $0.30 \pm 0.30$ \\ \hline
4  & $0.12 \pm 0.24$ & $0.18 \pm 0.27$ & $0.20 \pm 0.25$ & $0.22 \pm 0.27$ & $0.23 \pm 0.26$ & $0.26 \pm 0.26$ & $0.26 \pm 0.26$ & $0.27 \pm 0.28$ & $0.28 \pm 0.29$ \\ \hline
5  & $0.11 \pm 0.23$ & $0.17 \pm 0.24$ & $0.19 \pm 0.26$ & $0.21 \pm 0.26$ & $0.21 \pm 0.24$ & $0.23 \pm 0.26$ & $0.23 \pm 0.24$ & $0.23 \pm 0.24$ & $0.25 \pm 0.25$ \\ \hline
6  & $0.13 \pm 0.25$ & $0.14 \pm 0.22$ & $0.18 \pm 0.23$ & $0.20 \pm 0.24$ & $0.20 \pm 0.23$ & $0.21 \pm 0.24$ & $0.22 \pm 0.22$ & $0.22 \pm 0.25$ & $0.23 \pm 0.24$ \\ \hline
7  & $0.10 \pm 0.21$ & $0.14 \pm 0.22$ & $0.16 \pm 0.22$ & $0.18 \pm 0.22$ & $0.18 \pm 0.21$ & $0.20 \pm 0.22$ & $0.20 \pm 0.21$ & $0.21 \pm 0.21$ & $0.21 \pm 0.22$ \\ \hline
8  & $0.10 \pm 0.20$ & $0.14 \pm 0.22$ & $0.15 \pm 0.21$ & $0.17 \pm 0.21$ & $0.18 \pm 0.23$ & $0.19 \pm 0.22$ & $0.19 \pm 0.21$ & $0.21 \pm 0.22$ & $0.21 \pm 0.21$ \\ \hline
9  & $0.09 \pm 0.18$ & $0.12 \pm 0.20$ & $0.14 \pm 0.19$ & $0.16 \pm 0.20$ & $0.17 \pm 0.20$ & $0.18 \pm 0.20$ & $0.18 \pm 0.20$ & $0.19 \pm 0.21$ & $0.19 \pm 0.20$ \\ \hline
10 & $0.09 \pm 0.18$ & $0.12 \pm 0.19$ & $0.14 \pm 0.19$ & $0.14 \pm 0.19$ & $0.15 \pm 0.19$ & $0.16 \pm 0.19$ & $0.17 \pm 0.20$ & $0.18 \pm 0.20$ & $0.18 \pm 0.19$ \\ \hline
\end{tabular}
}
\caption{Mean regret $\pm$ standard deviation of the Thompson Sampling strategy on the Google Reviews data for restaurants.}
\label{tab:google-thompson-std}
\end{table}

\begin{table}[h]
\centering
\resizebox{\textwidth}{!}{
\begin{tabular}{|c|c|c|c|c|c|c|c|c|c|}
\hline
$m,n_d$ & 2 & 3 & 4 & 5 & 6 & 7 & 8 & 9 & 10 \\ \hline
1  & $0.23 \pm 0.23$ & $0.29 \pm 0.21$ & $0.36 \pm 0.21$ & $0.37 \pm 0.19$ & $0.38 \pm 0.16$ & $0.40 \pm 0.17$ & $0.41 \pm 0.15$ & $0.42 \pm 0.14$ & $0.43 \pm 0.14$ \\ \hline
2  & $0.21 \pm 0.22$ & $0.30 \pm 0.22$ & $0.34 \pm 0.20$ & $0.36 \pm 0.17$ & $0.39 \pm 0.17$ & $0.40 \pm 0.16$ & $0.42 \pm 0.16$ & $0.42 \pm 0.15$ & $0.44 \pm 0.14$ \\ \hline
3  & $0.22 \pm 0.23$ & $0.29 \pm 0.21$ & $0.34 \pm 0.20$ & $0.35 \pm 0.18$ & $0.39 \pm 0.18$ & $0.40 \pm 0.16$ & $0.41 \pm 0.15$ & $0.42 \pm 0.14$ & $0.44 \pm 0.14$ \\ \hline
4  & $0.22 \pm 0.23$ & $0.28 \pm 0.20$ & $0.33 \pm 0.20$ & $0.36 \pm 0.17$ & $0.39 \pm 0.17$ & $0.40 \pm 0.16$ & $0.41 \pm 0.15$ & $0.43 \pm 0.15$ & $0.43 \pm 0.14$ \\ \hline
5  & $0.22 \pm 0.23$ & $0.30 \pm 0.22$ & $0.33 \pm 0.20$ & $0.37 \pm 0.18$ & $0.38 \pm 0.17$ & $0.40 \pm 0.16$ & $0.41 \pm 0.15$ & $0.43 \pm 0.15$ & $0.44 \pm 0.14$ \\ \hline
6  & $0.21 \pm 0.23$ & $0.29 \pm 0.22$ & $0.33 \pm 0.20$ & $0.36 \pm 0.18$ & $0.37 \pm 0.16$ & $0.38 \pm 0.15$ & $0.42 \pm 0.16$ & $0.41 \pm 0.15$ & $0.43 \pm 0.14$ \\ \hline
7  & $0.22 \pm 0.24$ & $0.29 \pm 0.21$ & $0.35 \pm 0.21$ & $0.37 \pm 0.19$ & $0.39 \pm 0.17$ & $0.40 \pm 0.16$ & $0.41 \pm 0.15$ & $0.43 \pm 0.15$ & $0.43 \pm 0.14$ \\ \hline
8  & $0.20 \pm 0.21$ & $0.30 \pm 0.22$ & $0.34 \pm 0.20$ & $0.37 \pm 0.18$ & $0.38 \pm 0.18$ & $0.40 \pm 0.16$ & $0.41 \pm 0.16$ & $0.42 \pm 0.15$ & $0.44 \pm 0.14$ \\ \hline
9  & $0.22 \pm 0.22$ & $0.29 \pm 0.22$ & $0.34 \pm 0.20$ & $0.36 \pm 0.18$ & $0.39 \pm 0.17$ & $0.39 \pm 0.15$ & $0.40 \pm 0.15$ & $0.43 \pm 0.15$ & $0.43 \pm 0.13$ \\ \hline
10 & $0.22 \pm 0.25$ & $0.29 \pm 0.21$ & $0.34 \pm 0.20$ & $0.37 \pm 0.19$ & $0.38 \pm 0.17$ & $0.40 \pm 0.16$ & $0.41 \pm 0.15$ & $0.43 \pm 0.15$ & $0.43 \pm 0.14$ \\ \hline
\end{tabular}
}
\caption{Mean regret $\pm$ standard deviation of the uniform strategy on the Google Reviews data for restaurants.}
\label{tab:google-uniform-std}
\end{table}

All simulations were run locally on a 2019 15-inch MacBook Pro with a 2.4 GHz 8-core Intel Core i9 processor and 32 GB of RAM. The full set of 135,000 simulations for the reported grid completed in approximately 10--15 minutes.

\end{document}